\documentclass[10pt,a4paper]{article}
\usepackage[utf8]{inputenc}
\usepackage[english]{babel}
\usepackage{amsmath}
\usepackage{amsfonts}
\usepackage{amssymb}
\usepackage{amsthm}
\usepackage{cite}
\usepackage{graphicx}
\usepackage[left=2cm,right=2cm,top=2cm,bottom=2cm]{geometry}
\usepackage{hyperref}
\usepackage{comment}
\usepackage{authblk}

\hypersetup{
    colorlinks=true,
    linkcolor=blue,
    filecolor=magenta,      
    urlcolor=cyan,
    pdfpagemode=FullScreen,
}

\newtheorem{theorem}{Theorem}

\newtheorem{lemma}{Lemma}

\newtheorem{definition}{Definition}
\newtheorem{example}{Example}
\newtheorem{remark}{Remark}

\providecommand{\keywords}[1]{%
\small\textbf{\textit{Keywords:}} #1%
}

\title{Hamiltonian reduction from particular integrals}
\author[1]{R. Azuaje}
\affil[1]{Department of Physics, Faculty of Nuclear Sciences and Physical Engineering, Czech Technical University in Prague. Břehová 7, 115 19 Praha 1, Czech Republic}
\author[2]{A. M. Escobar-Ruiz}
\affil[2]{Departamento de F\'{i}sica, Universidad Aut\'onoma Metropolitana Unidad Iztapalapa\\
San Rafael Atlixco 186, 09340, Ciudad de M\'exico, M\'exico}
\author[3]{I. Gutierrez-Sagredo}
\affil[3]{Departamento de Matem\'aticas y Computaci\'on, Universidad de Burgos, 
09001 Burgos, Spain}

\date{}

\begin{document}

\maketitle

\begin{abstract}
We develop a geometric reduction mechanism generated by systems of particular integrals, namely, families of functions whose time derivatives close linearly on the family. Their common zero set is dynamically invariant. In the Hamiltonian case, under a weak involution condition, the restricted dynamics is presymplectic, and its characteristic quotient carries a reduced Hamiltonian flow. This yields a direct bridge between particular integrals, presymplectic reduction, and lower-dimensional Hamiltonian dynamics, and leads to a Liouville-type notion of particular integrability. We illustrate the framework through mechanical examples and lift constructions, including variants of the Eisenhart lift.

\end{abstract}


\tableofcontents

\section{Introduction}

Conservation laws are one of the central organizing principles of Hamiltonian mechanics. In the classical Liouville--Arnold framework, sufficiently many independent first integrals in involution lead to invariant tori, action-angle variables, and integration by quadratures \cite{Arnold78,AKN2006,BBT2003}. Geometric mechanics has extended this picture through symplectic and Poisson reduction, bi-Hamiltonian structures, and modern approaches to integrability \cite{AMRC2008,LPV2013,MM84,KM96,TT2022}. Superintegrability and polynomial algebras of constants of motion provide another important source of explicitly solvable mechanical systems \cite{miller2013classical,tempesta2004superintegrability,marquette2010superintegrability}.

The present work concerns a weaker but geometrically robust form of conservation. A quantity may fail to be globally conserved and yet become conserved after the dynamics is restricted to a suitable invariant subset of phase space. Such quantities are known as particular integrals. They were introduced in the context of particular integrability and quasi-exact solvability in \cite{Turbiner2013}, and later developed for classical Hamiltonian systems in \cite{escobar2024particular}. Extensions to cosymplectic, contact, and cocontact Hamiltonian systems were studied in \cite{azuaje2025particular}. A complementary viewpoint is provided by Lie integrability by quadratures, where restricted equations can be solved through solvable Lie algebras of symmetries \cite{CFGR2015,azuaje2024lie}.

The main point of this paper is that particular integrals naturally generate a Hamiltonian reduction mechanism. We first generalize the usual scalar condition
\[
\dot f=af
\]
to a system of functions
\[
\dot f_i=a_i^j f_j.
\]
The common zero set of such a family is dynamically invariant. In the Hamiltonian case, if the functions are in particular involution (a weak involution condition), the invariant submanifold inherits a closed two-form of constant rank. The restricted dynamics is therefore presymplectic, and quotienting the characteristic distribution yields a genuine Hamiltonian system with fewer degrees of freedom. This connects particular integrability with presymplectic reduction and constrained Hamiltonian dynamics \cite{carinena1985canonical,grabowska2023reductions,libermann2012symplectic}.

This viewpoint clarifies the distinction between different notions of integrability on restricted dynamics. The particular integrability studied in \cite{escobar2024particular} is naturally formulated in terms of Lie integrability by quadratures: the restricted equations need not be Hamiltonian. By contrast, when the restricted dynamics admits a presymplectic quotient, Liouville integrability becomes meaningful on the reduced symplectic phase space. This motivates the notion of \emph{particular Liouville integrability}: a Hamiltonian system is particularly Liouville integrable when the dynamics restricted by particular integrals and projected along the characteristic distribution becomes a completely Liouville integrable Hamiltonian system. This notion is related to weak and restricted forms of integrability, where complete integrability holds only on selected energy or momentum levels \cite{rosquist1995invariants,pucacco2004non}.

The construction also fits naturally within the broader geometric study of symmetries and conservation laws. Noether-type correspondences, canonoid transformations, and scaling symmetries show that conservation or reduction can arise from structures weaker than canonical symmetries \cite{KS2011,CR88,RS2015,azuaje2024scaling}. In this sense, particular integrals provide another mechanism by which non-global conservation data can organize Hamiltonian dynamics.

A second goal of the paper is to show that this mechanism appears naturally in lift constructions. The Eisenhart lift realizes natural Hamiltonian trajectories as projections of geodesic dynamics on an extended configuration space. More generally, auxiliary lifts can turn first integrals of an original Hamiltonian system into systems of particular integrals of a lifted system. The original dynamics is recovered by restricting to the invariant subset defined by the auxiliary momenta and projecting away the auxiliary variables. This perspective connects particular integrals with geometric lifts, mechanical Hamiltonian systems, and polynomial invariants of natural systems \cite{godbillon1969geometrie,maciejewski2004darboux,hietarinta1987direct}.

The paper is organized as follows. Section 2 introduces systems of particular integrals and proves that their common zero level sets are dynamically invariant. Section 3 specializes in Hamiltonian systems and shows that involutive systems of particular integrals define invariant presymplectic submanifolds whose characteristic quotients carry reduced Hamiltonian dynamics; this leads to the definition of particular Liouville integrability. Section 4 applies the construction to lifts of natural Hamiltonian systems, including scalar, diagonal, and non-diagonal auxiliary lifts, as well as systems with electromagnetic coupling. We close with a short outlook on polynomial Hamiltonians.

\section{Particular integrals and invariant submanifolds}

The notion of particular integral was introduced in \cite{Turbiner2013}, namely, \textbf{a particular integral for a Hamiltonian system is a smooth real function on phase space that is conserved possibly only for certain trajectories of the system}. Of course, every constant of motion is a particular integral. The notion of particular integral has been extensively studied in relation to the reduction of the equations of motion and with the notion of particular integrability \cite{escobar2024particular,azuaje2025particular}; namely, the integrability by quadratures of the Hamilton equations of motion on a part of the phase space (on a dynamically invariant submanifold determined by independent particular integrals). The concept of particular superintegrability, as an extension of the notion of superintegrability, has also received systematic attention \cite{turbiner2020particular,turbiner2021superintegrability,escobar2026nonlinear}.

Let us see a relevant example.
\begin{example}
\label{excfp}
Let us consider the classical central-force problem in $\mathbb{R}^3$.
\begin{equation}
\label{hc}
    {\rm H}(x,y,z,p_{x},p_{y},p_{z}) \ = \ \frac{1}{2\,m}\,\big(\,p_x^2 \ + \ p_y^2 \ + \ p_z^2 \,\big)  \ + \ V(r).
\end{equation}
Making a canonical transformation to the non-orthogonal coordinate system $(r,\,\phi,\,\rho)=(Q^1,Q^2,Q^3)$
\[
r= \sqrt{x^2+y^2+z^2} \quad , \qquad  \phi= \tan^{-1}\bigg(\frac{y}{x}\bigg) \quad , \qquad \rho= \sqrt{x^2+y^2} \ ,
\]
the Hamiltonian (\ref{hc}) becomes
\begin{equation*}
\label{hc2}
  K \ \equiv \ {\rm H}(Q,P) \ = \ \frac{1}{2\,m}\,\bigg(\,p_{\rho}^2  \ + \ p_{r}^2  \ + \ \frac{2\,\rho}{r}\, p_{\rho}\, p_{r} \ + \ \frac{1}{\rho^2}\, p_{\phi}^2\,\bigg) \ + \ V(r).
\end{equation*}
We find that $p_{\phi}$ is a constant of motion and  $p_\rho$ is conserved only on a dynamically invariant submanifold, namely, that defined by the level subset $p_{\phi}=0,p_{\rho}=0$, indeed
\[
\lbrace p_{\rho},\,K\rbrace \ = \ - \bigg(\frac{1}{m\,r}\,p_r\bigg)\,p_\rho \ + \ \bigg(\frac{1}{m\,\rho^3}\,p_{\phi}\,\bigg)\,p_{\phi}.
\]
\end{example}

It is well known that the existence of a constant of motion for a classical mechanical system enables us to restrict the dynamics to a dynamically invariant submanifold. Namely, if $f$ is a constant of motion of a classical mechanical system, then we have that for each $\alpha\in\mathbb{R}$ the level set 
\[ 
M^{(\alpha)}_{f} \ = \ \lbrace\, x: \ f(x)\,=\,\alpha \, \rbrace \ ,
\]
defines a (smooth) submanifold of the phase space of codimension $1$ for regular values of $f$ \cite{Lee2012}. Moreover, it is invariant under the dynamics. In the particular case of the previous example, we have that the level subset defined by $p_{\phi}=0$ is a dynamically invariant submanifold of the given Hamiltonian system; in such restricted dynamics, the evolution of $p_{\rho}$ is determined by
\begin{equation*}
\dot{p}_{\rho}=- \bigg(\frac{1}{m\,r}\,p_r\bigg)\,p_\rho .
\end{equation*}
Following this pattern, in \cite{escobar2024particular,azuaje2025particular} a particular integral is defined as a smooth function $f$ on the phase space such that $\dot{f}=a\,f$ for some smooth function $a$ on the phase space. This definition does not include more general cases, as shown in the following example.

\begin{example}
\label{exfourbody}
 Consider a classical system of four particles, moving in the Euclidean plane $\mathbb{R}^{2}$, with equal masses ($m_1 = m_2 = m_3 = m_4 = 1$) and subjected to a quadratic pairwise interaction potential; the Hamiltonian function is \cite{escobar2025four}
\begin{equation*}
H(\textbf{r}_{1},\textbf{r}_{2},\textbf{r}_{3},\textbf{r}_{4},\textbf{p}_{1},\textbf{p}_{2},\textbf{p}_{3},\textbf{p}_{4})=\frac{1}{2}\left(\textbf{p}_{1}^{2}+\textbf{p}_{2}^{2}+\textbf{p}_{3}^{2}+\textbf{p}_{4}^{2}\right)+\frac{1}{2}w^{2}\left(r_{12}^{2}+r_{23}^{2}+r_{34}^{2}+r_{14}^{2}-\frac{1}{2}\left(r_{13}^{2}+r_{24}^{2}\right)\right),
\end{equation*}
with $\textbf{r}_{ij}=\textbf{r}_{i}-\textbf{r}_{j}$ and $r_{ij}=\vert \textbf{r}_{ij}\vert$.
Now we introduce (vector) Jacobi coordinates $(\textbf{S}_{1},\textbf{S}_{2},\textbf{S}_{3},\textbf{S}_{4})$ on the configuration space, namely
\begin{equation*}
\left\lbrace \begin{array}{c}
\textbf{S}_{1}=\frac{1}{2}\left(\textbf{r}_{2}-\textbf{r}_{1}\right)\\
\textbf{S}_{2}=\sqrt{\frac{2}{3}}\left(\textbf{r}_{3}-\frac{1}{2}\left(\textbf{r}_{1}+\textbf{r}_{2}\right)\right)\\
\textbf{S}_{3}=\sqrt{\frac{3}{4}}\left(\textbf{r}_{4}-\frac{1}{3}\left(\textbf{r}_{1}+\textbf{r}_{2}+\textbf{r}_{3}\right)\right)\\
\textbf{S}_{4}=\frac{1}{4}\left(\textbf{r}_{1}+\textbf{r}_{2}+\textbf{r}_{3}+\textbf{r}_{4}\right);
\end{array} \right.;
\end{equation*}
then the conjugate momenta $\textbf{P}_{1},\textbf{P}_{2},\textbf{P}_{3},\textbf{P}_{4}$ are
\begin{equation*}
\left\lbrace \begin{array}{c}
\textbf{P}_{1}=\frac{1}{2}\left(\textbf{p}_{2}-\textbf{p}_{1}\right)\\
\textbf{P}_{2}=\sqrt{\frac{2}{3}}\left(\textbf{p}_{3}-\frac{1}{2}\left(\textbf{p}_{1}+\textbf{p}_{2}\right)\right)\\
\textbf{P}_{3}=\sqrt{\frac{3}{4}}\left(\textbf{p}_{4}-\frac{1}{3}\left(\textbf{p}_{1}+\textbf{p}_{2}+\textbf{p}_{3}\right)\right)\\
\textbf{P}_{4}=\textbf{p}_{1}+\textbf{p}_{2}+\textbf{p}_{3}+\textbf{p}_{4}
\end{array} \right. \\.
\end{equation*}

We have
\begin{equation*}
\begin{split}
H(\textbf{S}_{1},\textbf{S}_{2},\textbf{S}_{3},\textbf{S}_{4},\textbf{P}_{1},\textbf{P}_{2},\textbf{P}_{3},\textbf{P}_{4})&=\frac{1}{2}\left(\textbf{P}_{1}^{2}+\textbf{P}_{2}^{2}+\textbf{P}_{3}^{2}\right)+\frac{1}{8}\textbf{P}_{4}^{2}\\
&\quad+\frac{1}{4}w^{2}\left(5\textbf{S}_{1}^{2}+3\textbf{S}_{2}^{2}+4\textbf{S}_{3}^{2}-2\sqrt{3}\textbf{S}_{1}\textbf{S}_{2}+2\sqrt{6}\textbf{S}_{1}\textbf{S}_{3}-2\sqrt{2}\textbf{S}_{2}\textbf{S}_{3}\right).
\end{split}
\end{equation*}
Observe that $\textbf{S}_{4}$ is a cyclic coordinate, so $\textbf{P}_{4}$ is a (vector) constant of motion; in addition 
\begin{equation*}
\dot{S}_{4}=\frac{1}{4}\textbf{P}_{4},
\end{equation*}
so when $\textbf{P}_{4}=0$ (the level set of the phase space determined by $\textbf{P}_{4}=0$ is a dynamically invariant submanifold) we have that $\textbf{S}_{4}$ is a conserved quantity, i.e., $\textbf{S}_{4}$ is a (vector) particular integral (formally speaking, the components of $\textbf{S}_{4}$ are particular integrals).

We consider the dynamics restricted to $\textbf{P}_{4}=0$ (therefore $\textbf{S}_{4}$ is a real vector constant under each trajectory), the reduced phase space $M$ of dimension $12$ has vector canonical coordinates $(\textbf{S}_{1},\textbf{S}_{2},\textbf{S}_{3},\textbf{P}_{1},\textbf{P}_{2},\textbf{P}_{3})$  and the reduced Hamiltonian is
\begin{equation*}
H|_{\textbf{P}_{4}=0}(\textbf{S}_{1},\textbf{S}_{2},\textbf{S}_{3},\textbf{P}_{1},\textbf{P}_{2},\textbf{P}_{3})=\frac{1}{2}\left(\textbf{P}_{1}^{2}+\textbf{P}_{2}^{2}+\textbf{P}_{3}^{2}\right)+\frac{1}{4}w^{2}\left(5\textbf{S}_{1}^{2}+3\textbf{S}_{2}^{2}+4\textbf{S}_{3}^{2}-2\sqrt{3}\textbf{S}_{1}\textbf{S}_{2}+2\sqrt{6}\textbf{S}_{1}\textbf{S}_{3}-2\sqrt{2}\textbf{S}_{2}\textbf{S}_{3}\right).
\end{equation*}
$H|_{\textbf{p}_{4}=0}$ coincides with the Hamiltonian function for the relative motion to the center of mass, namely $H_{rel}=H|_{\textbf{P}_{4}=0,\textbf{S}_{4}=0}$ (the restrictions $\textbf{P}_{4}=0,\textbf{S}_{4}=0$ define the center of mass frame), so we shall denote $H|_{\textbf{p}_{4}=0}=H_{rel}$, however, the dynamics is not restricted to $\textbf{S}_{4}=0$.

Now, let us consider the scalar quantities $rp_{13}=\textbf{r}_{13}\cdot \textbf{p}_{13}$ and $l_{13}=\textbf{p}_ {13}^{2}-w^{2}\textbf{r}_{13}^{2}$,
\begin{equation*}
rp_{13}=\frac{1}{2}(\textbf{P}_{1}+\sqrt{3}\textbf{P}_{2})(\textbf{S}_{1}+\sqrt{3}\textbf{S}_{2})
\end{equation*}
and
\begin{equation*}
l_{13}=\frac{1}{2}\left(\textbf{P}_{1}^2+2\sqrt{3}\textbf{P}_{1}\textbf{P}_{2}+3\textbf{P}_{2}^2-(\textbf{S}_{1}^2+2\sqrt{3}\textbf{S}_{1}\textbf{S}_{2}+3\textbf{S}_{2}^2)w^2\right).
\end{equation*}
We have
\begin{equation*}
\dot{rp}_{13}=l_{13}\quad\textit{and}\quad\dot{l}_{13}=-4w^{2}rp_{13},
\end{equation*}
so we have that $rp_{13}$ and $l_{13}$ are particular integrals of the reduced Hamiltonian (therefore, of the original Hamiltonian); indeed, by taking simultaneously $rp_{13}=0$ and $l_{13}=0$, we have that $rp_{13}$ and $l_{13}$ are constant on each trajectory. 
\end{example}

In the previous example, the functions $S_4$, $rp_{13}$ and $l_{13}$ do not satisfy the scalar condition $\dot f=af$. Rather, they illustrate a more general mechanism: invariant submanifolds may arise as common zero sets of several functions whose time derivatives close linearly on the same set of functions. The following lemma makes this observation precise and generalizes Lemma 1 of \cite{escobar2024particular,azuaje2025particular}.

\begin{lemma}
\label{lemma1}
Let $X$ be the dynamical vector field of a classical mechanical system on a smooth phase space $M$. Let $f_1,\ldots,f_k\in C^\infty(M)$ satisfy
\begin{equation}
\label{eqPI}
X(f_i)=a_i^j f_j,\qquad i=1,\ldots,k,
\end{equation}
where summation over the repeated index $j$ is understood and $a_i^j\in C^\infty(M)$. Suppose that $f_1,\ldots,f_k$ are functionally independent on
\begin{equation*}
\label{Mf}
M_f=\{x\in M:\ f_1(x)=0,\ldots,f_k(x)=0\}.
\end{equation*}
Then $M_f$ is a dynamically invariant embedded submanifold of $M$ of codimension $k$.
\end{lemma}

\begin{proof}
Let
\[
F=(f_1,\ldots,f_k):M\longrightarrow \mathbb{R}^k .
\]
The functional independence assumption means that $F$ has rank $k$ at every point of $M_f=F^{-1}(0)$. Hence $M_f$ is an embedded submanifold of codimension $k$ by the regular level set theorem \cite{Lee2012}.

Let $\gamma:I\to M$ be an integral curve of $X$ such that $\gamma(t_0)\in M_f$. Set
\[
u_i(t)=f_i(\gamma(t)),\qquad i=1,\ldots,k.
\]
Then \eqref{eqPI} gives
\[
\dot u_i(t)=a_i^j(\gamma(t))u_j(t),
\]
or, equivalently,
\[
\dot u(t)=A(t)u(t),\qquad A(t)=\big(a_i^j(\gamma(t))\big).
\]
Since $A(t)$ is smooth and $u(t_0)=0$, uniqueness for linear ordinary differential equations implies $u(t)\equiv 0$ on $I$. Therefore $f_i(\gamma(t))=0$ holds for all $i=1,\ldots,k$ and all $t\in I$, so $\gamma(t)\in M_f$ holds for all $t\in I$. Thus $M_f$ is dynamically invariant.
\end{proof}

Following Lemma \ref{lemma1}, we introduce the following generalized notion.

\begin{definition}
\label{particularintegrals}
Let $X$ be the dynamical vector field of a classical mechanical system on a smooth phase space $M$. A family $f_1,\ldots,f_k\in C^\infty(M)$ is called a system of particular integrals if there exist functions $a_i^j\in C^\infty(M)$ such that
\begin{equation*}
X(f_i)=a_i^j f_j,\qquad i=1,\ldots,k.
\end{equation*}
\end{definition}

For $k=1$, this reduces to the usual condition $\dot f=af$. If, for a fixed $i$, one has $a_i^j=0$ for all $j$, then $f_i$ is a constant of motion. In particular, if all coefficients $a_i^j$ vanish, then $f_1,\ldots,f_k$ are constants of motion.
It is clear that a function $f$ that satisfies $\dot{f}=af$ is a particular integral. 

Lemma \ref{lemma1} shows that systems of particular integrals define invariant submanifolds of phase space. Thus, they provide a natural mechanism for restricting the dynamics to lower-dimensional equations of motion.

\section{Reduction and particular Liouville integrability}

In \cite{escobar2024particular}, it is shown that in the case of classical Hamiltonian systems, the existence of a particular integral ($\dot{f}=af$) allows us to find trajectories of the system by solving a reduced system of Hamilton's equations of motion. Concretely, given a particular integral, we restrict the dynamics to the dynamically invariant submanifold of codimension $1$ defined by the zero level set of the given particular integral. The equations of motion of such restricted dynamics can be solved by first solving a system of Hamilton's equations of motion of a Hamiltonian system with the number of degrees of freedom reduced by $1$, and afterwards integrating a sole differential equation.  The dynamics defined by the reduced Hamiltonian system is the projection of the dynamics on the invariant submanifold into the submanifold defined by the level sets of the remaining coordinate to be integrated;  in general, it is not contained in the original dynamics (for details, see \cite{escobar2024particular}). 

In this section, we show that, for a given classical Hamiltonian system, a regular system of particular integrals defines a Hamiltonian reduction whenever its common zero set is coisotropic.  The formulation is independent of local canonical coordinates.

For this section, let $(M,\omega)$ be a symplectic manifold of dimension $2n$ and let $H\in C^\infty(M)$. The Hamiltonian vector field $X_{H}$ for $H$ is determined by
\begin{equation*}
i_{X_H}\omega= dH,
\end{equation*}
and the time-evolution of a function $f\in C^{\infty}(M)$ is given by
\begin{equation*}
\dot f=X_H(f)=\{f,H\}.
\end{equation*}

\subsection{Projected Hamiltonian Dynamics}

We now show that a system of particular integrals in particular involution (weak involution) produces Hamiltonian dynamics after restriction and projection.

\subsubsection{Global geometric description}

We start with a brief review of the geometry of presymplectic manifolds and Hamiltonian dynamics on presymplectic manifolds (for further details, see \cite{gotay1978presymplectic,carinena1985canonical,echeverria1999reduction,bursztyn2013brief}).

A presymplectic manifold is a smooth manifold $N$ endowed with a closed 2-form $\Omega$ (called a presymplectic form). 

\begin{remark}
For our aim, we consider regular presymplectic manifolds, i.e., presymplectic manifolds with presymplectic forms of constant rank.
\end{remark}
Given a regular presymplectic manifold $(N,\Omega)$, around any point $x\in N$, there exist local coordinates\\ $(q^{1},\ldots,q^{s},p_{1},\ldots,p_{s},z^{1},\ldots,z^{r})$, called canonical coordinates, such that $\Omega=dq{i}\wedge dp_{i}$, where $dim(N)=2s+r$ and $2s=rank(\Omega)$ (see Darboux Theorem for presymplectic manifolds in \cite{gracia2024darboux}).

Let $(N,\Omega)$ be a presymplectic manifold of dimension $2s+r$ with $rank(\Omega)=2s$. The distribution 
$$\ker(\Omega)=\lbrace X\in\mathfrak{X}(N):\ i_{X}\Omega=0 \rbrace$$
is called the characteristic distribution of $\Omega$, it is an integrable distribution of rank $r$. In local canonical coordinates, it is generated by the coordinate vector fields $\lbrace \frac{\partial}{\partial z^{1}},\ldots,\frac{\partial}{\partial z^{r}}\rbrace$. 

A function $f\in C^{\infty}(N)$ is said to be admissible for defining Hamiltonian dynamics on $(N,\Omega)$ (shortened to just admissible) if there exists a vector field $X_{f}$ on $N$ such that 
\begin{equation}
\label{eqHamiltonvector}
i_{X_{f}}\Omega=df. 
\end{equation}
Due to the degeneracy of $\Omega$, $X_{f}$ is generally not unique; in fact, if $X$ is such that $i_{X}\Omega=df$, then $i_{X+Z}\Omega=df$ for every $Z\in\ker(\Omega)$. The set of all vector fields on $N$ satisfying (\ref{eqHamiltonvector}) is denoted by $Ham_{f}(N)$.

It has been shown that a function $f\in C^{\infty}(N)$ is admissible if and only if $i_{Z}df=0$ for all $Z\in\ker(\Omega)$. In local terms, a function $f$ is admissible if and only if its local expression in canonical coordinates $(q,p,z)$ satisfies $\frac{\partial f}{\partial z^{j}}=0$ for $j=1,\ldots,r.$ Each $X\in Ham_{f}(N)$ has the local expression
\begin{equation*}
X=\frac{\partial f}{\partial p_{i}}\frac{\partial}{\partial q^{i}}-\frac{\partial f}{\partial q^{i}}\frac{\partial}{\partial p_{i}}+A^{j}\frac{\partial}{\partial z^{j}},
\end{equation*}
for some locally defined function $A^{1},\ldots,A^{r}$.

Let $H\in C^{\infty}(N)$ be an admissible function on $(N,\Omega)$. Since the Hamiltonian vector field $X_{H}$ is not unique, for each point $x\in N$ we have a set of integral curves passing through it (one for each $X\in Ham_{H}(N)$). They are called gauge equivalent trajectories and are supposed to represent the same physical state of the system. In order to represent each physical state by only one trajectory of a mechanical system, the so-called gauge reduction procedure is applied \cite{echeverria1999reduction}. It consists of constructing a Hamiltonian system on a symplectic manifold such that each trajectory represents a class of gauge equivalent trajectories of the original presymplectic Hamiltonian dynamics. For this, it is assumed that the quotient space
\begin{equation*}
\overline{N}=\frac{N}{ker(\Omega)}
\end{equation*}
is a differentiable manifold (it suffices to assume that the foliation of the distribution $\ker(\Omega)$ is simple \cite{libermann2012symplectic,grabowska2023reductions}); it has dimension $2s$ and the natural projection $\pi:N\longrightarrow\overline{N}$ is a submersion that endows it with a symplectic form $\omega$ such that $\pi^{*}\omega=\Omega$. There is one and only one vector field $\overline{X}_{H}$ on $\overline{N}$ such that 
\begin{equation*}
\pi_{*}X=\overline{X}_{H}, \ \forall X\in Ham_{H}(N).
\end{equation*}
In fact, $\overline{X}_{H}$ is the Hamiltonian vector field for the function $\overline{H}\in C^{\infty}(\overline{N})$ such that $\pi^{*}\overline{H}=H$, i.e.,
\begin{equation*}
i_{\overline{X}_{H}}\omega=d\overline{H}.
\end{equation*}
The system $(\overline{N},\omega,\overline{H})$ is called the reduced gauge-free Hamiltonian system for $(N,\Omega,H)$.

Now we state the main result of this paper.
\begin{theorem}
\label{maintheorem}
Let $f_1,\ldots,f_k\in C^\infty(M)$, $k<n$, be a system of particular integrals for $(M,\omega,H)$, i.e., 
\begin{equation*}
\{f_i,H\}=a_{i}^j f_j,
\qquad i=1,\ldots,k,
\end{equation*}
for suitable smooth functions $a_i{}^j$; such that they are in particular involution, i.e.,
\begin{equation}
\{f_i,f_j\}=c_{ij}{}^\ell f_\ell,
\qquad i,j=1,\ldots,k,
\label{eq:weak-involution}
\end{equation}
where the $c_{ij}{}^\ell$ are smooth functions. 

Assume that $f_1,\ldots,f_k$ are functionally independent on the level set
\begin{equation*}
M_f=\{x\in M:\ f_1(x)=0,\ldots,f_k(x)=0\}.
\end{equation*}

Then, $M_f$ is a presymplectic manifold of dimension $2n-k$ with the presymplectic form $\Omega=\iota^{*}\omega$, where $\iota\colon M_f \hookrightarrow M$ is the inclusion map, its characteristic distribution is
\begin{equation*}
\ker\Omega=\operatorname{span}\bigl\{X_{f_1}|_{M_f},\ldots,X_{f_k}|_{M_f}\bigr\};
\end{equation*}
and the restricted Hamiltonian $H_{M_f}=H|_{M_f}$ is constant along the characteristic leaves.

If in addition, the characteristic foliation $\ker(\Omega)$ is simple and its leaf space
\[
  \overline{M}=M_f/\ker(\Omega)
\]
is a smooth manifold for which the quotient map
\(\pi:M_f\to \overline{M}\) is a surjective submersion, then there exists a unique symplectic form $\overline{\omega}$ and a unique function
$\overline{H}\in C^\infty(\overline{M})$ such that
\begin{equation}
  \pi^*\overline{\omega}=\Omega,
  \qquad
  \pi^*\overline{H}=H_{M_f}.
  \label{eq:reduced-form-and-Hamiltonian}
\end{equation}
Moreover, $X_H|_{M_f}$ is projectable and
\begin{equation}
  \pi_*(X_H|_{M_f})=X_{\overline{H}},
  \qquad
  \iota_{X_{\overline{H}}}\overline{\omega}
  =\mathrm d\overline{H}.
  \label{eq:projected-Hamiltonian-vector-field}
\end{equation}
In particular, $\overline{M}$ has dimension $2(n-k)$.
\end{theorem}
\begin{proof}
Let
\[
F=(f_1,\ldots,f_k)\colon M\longrightarrow\mathbb{R}^k.
\]
By the functional-independence hypothesis along $M_f$, we have that $0$ is a regular value of $F$, so, from Lemma \ref{lemma1} we have that $M_f$ is a dynamically invariant submanifold of dimension $2n-k$.

Since $0$ is a regular value of $F$, at every $x\in M_f$ one has
\[
  T_xM_f=\bigcap_{i=1}^k\ker(df_i(x))
\]
and therefore
\begin{equation*}
(T_xM_f)^\omega=\operatorname{span}\bigl\{X_{f_1}(x),\ldots,X_{f_k}(x)\bigr\}.
\label{eq:symplectic-orthogonal}
\end{equation*}
For every $i,j$, we have
\[
  X_{f_i}(f_j)|_{M_f}=\{f_j,f_i\}|_{M_f}=0
\]
by \eqref{eq:weak-involution}; so, each $X_{f_i}$ is tangent to $M_f$, consequently, $M_f$ is coisotropic. Furthermore,
\[
  \ker\Omega_x=T_xM_f\cap(T_xM_f)^\omega=(T_xM_f)^\omega.
\]
The independence of the $df_i$ implies that the vectors $X_{f_i}|_{M_f}$ are independent, so $\ker\Omega$ has constant rank $k$.  Since $\dim M_f=2n-k$, it follows that $rank\Omega=2n-2k$.

For every characteristic vector field $X_{f_i}|_{M_f}$, we have
\[
dH_{M_f}(X_{f_i}|_{M_f})
=X_{f_i}(H)|_{M_f}
=\{H,f_i\}|_{M_f}=0.
\]
Therefore, $H_{M_f}$ is constant on the leaves.

If the quotient is smooth, then there is a unique function $\overline{H}$ satisfying the second identity in
\eqref{eq:reduced-form-and-Hamiltonian}.

The form $\Omega$ is horizontal with respect to $\ker(\Omega)$ by definition of the characteristic distribution.  It is also invariant along characteristic vector fields: for $Z\in\ker\Omega$,
\[
\mathcal L_Z\Omega = d(\iota_Z\Omega)+\iota_Z d\Omega=0.
\]
Hence, $\Omega$ is basic and descends to the unique form
$\overline{\omega}$ in \eqref{eq:reduced-form-and-Hamiltonian}.  This form is closed and nondegenerate, and is therefore symplectic.

Finally, because $X_H$ is tangent to $M_f$,
\[
\iota_{X_H|_{M_f}}\Omega= dH_{M_f}.
\]
The flow of $X_H|_{M_f}$ preserves $\Omega$ and its kernel, so $X_H|_{M_f}$ is projectable. Pulling back the defining equation for its projection gives \eqref{eq:projected-Hamiltonian-vector-field}.  Uniqueness follows from the nondegeneracy of $\overline{\omega}$.
\end{proof}

\begin{remark}
\label{rem:standard-coisotropic-reduction}
The geometric part of Theorem \ref{maintheorem} is the standard coisotropic reduction
\cite{gotay1978presymplectic,echeverria1999reduction,deLeon_2024}.
The role of the particular integrals is to provide the invariance
condition directly from the dynamics. This identifies systems of particular integrals with regular generators of an invariant constraint ideal.
\end{remark}

\subsubsection{Local canonical coordinate description}
Consider a Hamiltonian system with $n$ degrees of freedom and a Hamiltonian function in canonical coordinates in phase space $M$ ($dim(M)=2n$) $H=H(q^{1},\ldots,q^{n},p_{1},\ldots,p_{n})$. 

Let $k\leq n$ and $f_{1},\cdots,f_{k}$ be functionally independent particular integrals in involution, i.e.,
\begin{equation*}
\dot{f}_{i} = \ a^{j}_{i}f_{j}
\end{equation*}
for some smooth functions $a^{j}_{i}$ on the phase space $M$, and
\begin{equation*}
\lbrace f_{l},f_{s}\rbrace=0.
\end{equation*}
Since $f_{1},\cdots,f_{k}$ are functionally independent and in involution, the Carathéodory-Jacobi-Lie theorem \cite{libermann2012symplectic,AKN2006,Lee2012} states that we can always find a (local) canonical transformation $(q,p)\mapsto (Q,P)$ such that the given functions are among the new canonical momenta, let us say without loss of generality, $P_{i}=f_{i}$ for $i=1,\ldots,k$ (the corresponding canonical transformation can be described by a generating function of the form $F=F_{2}(q^{1},\ldots,q^{n},f_{1},\ldots,f_{k},P_{k+1}\ldots,P_{n})-Q^{j}P_{j}$ \cite{GPS2002,landau1982mechanics,Calkin96}.) Hamilton's equations of motion in the new canonical coordinates $(Q,P)$ take the form
\begin{equation*}
\label{eqK}
\left\lbrace \begin{array}{c}
\dot{Q}^{1}= \frac{\partial K}{\partial P_{1}},\\ 
\vdots \\ 
\dot{Q}^{k}= \frac{\partial K}{\partial P_{k}},\\ 
\dot{Q}^{k+1}= \frac{\partial K}{\partial P_{k+1}},\\
\vdots\\
\dot{Q}^{n}= \frac{\partial K}{\partial P_{n}},\\ 
\dot{P}_{1}=\dot{f}_{1}=-\frac{\partial K}{\partial Q^{1}},\\ 
\vdots \\ 
\dot{P}_{k}=\dot{f}_{k}=-\frac{\partial K}{\partial Q^{k}},\\ 
\dot{P}_{k+1}= -\frac{\partial K}{\partial Q^{k+1}},\\
\vdots\\
\dot{P}_{n}= -\frac{\partial K}{\partial Q^{n}},
\end{array} \right.
\end{equation*}
with $K=H(Q,P)$ the new Hamiltonian function. Now we restrict the dynamics to $M_{f}=\lbrace x\in M: \ f_{1}(x)=0, \ldots,f_{k}(x)=0\rbrace$ (by Lemma \ref{lemma1} we have that $M_{f}$ is a dynamically invariant submanifold of $M$ of dimension $2n-k$); the equations of motion on the coordinates $Q^{1},\ldots,Q^{n},P_{k+1},\ldots,P_{n}$ on $M_{f}$ are
\begin{equation}
\label{reducedsystem}
\left\lbrace \begin{array}{c}
\dot{Q}^{1}= \frac{\partial K}{\partial P_{1}},\\ 
\vdots \\ 
\dot{Q}^{k}= \frac{\partial K}{\partial P_{k}},\\ 
\dot{Q}^{k+1}= \frac{\partial K}{\partial P_{k+1}},\\
\vdots\\
\dot{Q}^{n}= \frac{\partial K}{\partial P_{n}},\\ 
\dot{P}_{k+1}= -\frac{\partial K}{\partial Q^{k+1}},\\
\vdots\\
\dot{P}_{n}= -\frac{\partial K}{\partial Q^{n}}.
\end{array} \right.
\end{equation}

The local expression of the presymplectic form $\Omega$ on $M_{f}$ is
\begin{equation*}
\Omega=\sum_{j=k+1}^{n}dQ^{j}\wedge dP_{j}.
\end{equation*}
Since $M_{f}$ is, in general, a presymplectic manifold, which may have either even or odd dimension, the system (\ref{reducedsystem}) is not, in general, a system of Hamilton's equations of motion (even in the case when $M_{f}$ has even dimension, the inherited 2-form may be degenerate); however, since we take $P_{i}=f_{i}=0$ for $i=1,\ldots,k$, thus, $\dot{P}_{i}=-\frac{\partial K}{\partial Q^{i}}\big|_{{P_1=0,\ldots,P_k=0}}=0$, the function $K\big|_{M_{f}}$ does not depend on the coordinates $Q^{1},\ldots,Q_{k}$; therefore, the time evolution of each coordinate $Q^{k+1},\cdots,Q^{n},P_{k+1},\cdots,P_{n}$ in $M_{f}$ is independent of $Q^{1},\ldots,Q_{k}$. From the previous analysis, we can split the system (\ref{reducedsystem}) as follows.
\begin{equation}
\label{reducedhamsys}
\left\lbrace \begin{array}{c}
\dot{Q}^{k+1}= \frac{\partial K}{\partial P_{k+1}},\\
\vdots\\
\dot{Q}^{n}= \frac{\partial K}{\partial P_{n}},\\ 
\dot{P}_{k+1}= -\frac{\partial K}{\partial Q^{k+1}},\\
\vdots\\
\dot{P}_{n}= -\frac{\partial K}{\partial Q^{n}},
\end{array} \right.
\end{equation}
and
\begin{equation}
\label{inteq}
\left\lbrace \begin{array}{c}
\dot{Q}^{1}= \frac{\partial K}{\partial P_{1}},\\ 
\vdots \\ 
\dot{Q}^{k}= \frac{\partial K}{\partial P_{k}}.\\ 
\end{array} \right.
\end{equation}

We find that system (\ref{reducedhamsys}) is a system of Hamilton's equations of motion in coordinates $(Q^{k+1},\cdots,Q^{n},P_{k+1},\cdots,P_{n})$. In fact, we see that $(Q^{k+1},\cdots,Q^{n},P_{k+1},\cdots,P_{n})$ are local coordinates of the level set $M_{f,Q}=\lbrace x\in M_{f}:\ Q^{1}(x)=\alpha_{1},\ldots,Q^{k}(x)=\alpha_{k} \rbrace$ for any real values of $\alpha_{1},\ldots,\alpha_{k}$, and $K\big|_{M_{f,Q}}$ is the Hamiltonian function for the local Hamiltonian system on $M_{f,Q}$ with Hamilton's equations of motion (\ref{reducedhamsys}).

\begin{remark}
$M_{f,Q}$ is not a dynamically invariant set since $Q^{1},\ldots,Q^{k}$ are not necessarily constant along the trajectories of the system. 
\end{remark}

The dynamics of such a Hamiltonian system on $M_{f,Q}$ is the projection of the original dynamics into the submanifold $M_{f,Q}$ in the same sense as in \cite{escobar2024particular}, i.e., if $\pi:M_{f}\longrightarrow M_{f,Q}$ is the projection map from $M_{f}$ to the submanifold $M_{f,Q}$ defined in the canonical coordinates $(Q,P)$ by 
\begin{equation*}
\pi(Q^{1},\cdots,Q^{k},Q^{k+1},\ldots,Q^{n},P_{k+1},\cdots,P_{n})\ = \ (Q^{k+1},\cdots,Q^{n},P_{k+1},\cdots,P_{n}),
\end{equation*}
then the trajectories of the system on $M_{f,Q}$ are images by $\pi$ of the trajectories of the restricted system on $M_{f}$, i.e., $\gamma(t)=(\,Q^{k+1}(t),\cdots,Q^{n}(t),P_{k+1}(t),\cdots,P_{n}(t)\,)$ is a trajectory of the system on $M_{f,Q}$ if and only if there is a trajectory $\overline{\gamma(t)}=(Q^{1}(t),\cdots,Q^{k}(t),Q^{k+1}(t),\ldots,Q^{n}(t),P_{k+1}(t),\cdots,P_{n}(t))$ such that $\pi(\overline{\gamma(t)})=\gamma(t)$ for every $t$.

\begin{remark}
The system of equations of motion on $M_{f}$ (system (\ref{reducedsystem})) can be solved by first solving the system of Hamilton's equations of motion (\ref{reducedhamsys}) and then integrating system (\ref{inteq}): Rigorously speaking, to lift the dynamics from $M_{f,Q}$ to $M_{f}$, we consider trajectories $\gamma(t)=(Q^{k+1}(t),\cdots,Q^{n}(t),P_{k+1}(t),\cdots,P_{n}(t))$ of the system on $M_{f,Q}$ (solutions of system (\ref{reducedhamsys})) and construct trajectories $\overline{\gamma(t)}=(Q^{1}(t),\cdots,Q^{k}(t),Q^{k+1}(t),\ldots,Q^{n}(t),P_{k+1}(t),\cdots,P_{n}(t))$ on $M_{f}$ by integrating system (\ref{inteq}). 
\end{remark}

Under the regularity conditions of Theorem \ref{maintheorem}, $M_{f,Q}$ is the local representative of $\overline{M}$. In terms of the presymplectic structure, this projection procedure is known as the gauge reduction procedure. It consists of removing the redundancy of solutions of Hamiltonian vector fields for a given Hamiltonian function on a presymplectic manifold, and it ensures the existence of a unique Hamiltonian function on a reduced manifold related to the original Hamiltonian by means of a natural projection \cite{echeverria1999reduction}; which leads to a unique Hamiltonian vector field determining the global dynamics.

\begin{example}
Let us consider again the classical central-force problem in $\mathbb{R}^3$. Hamilton's equations of motion in the canonical coordinates $(r,\,\phi,\,\rho)$ (see example \ref{excfp}) are
\begin{equation*}
\label{eqmcp}
\left\lbrace \begin{array}{c}
 {\dot r}  \ = \ \frac{1}{m}\,p_r \ + \ \frac{1}{m}\frac{\rho}{r}\,p_\rho ,
\\ 
 {\dot p}_r \ = \ \frac{1}{m}\frac{\rho}{r^2}\, p_{\rho}\, p_{r} \ - \ V^\prime(r),
  \\ 
 {\dot \rho}  \ = \ \frac{1}{m}\,p_\rho \ + \ \frac{1}{m}\frac{\rho}{r}\,p_r,
 \\ 
 {\dot p}_\rho \ = \ -\frac{1}{m\,r}\, p_{\rho}\, p_{r} \ + \ \frac{1}{m\,\rho^3}\,p_\phi^2 ,
 \\ 
 {\dot \phi}  \ = \ \frac{1}{m\,\rho^2}\,p_\phi,
 \\ 
 {\dot p}_\phi  \ = \ 0 \ .
\end{array} \right.
\end{equation*}

We have that $p_{\phi},p_{\rho}$ are independent particular integrals in involution. By restricting the dynamics to $p_{\phi}=0,p_{\rho}=0$, the Hamilton equations of motion reduce to
\begin{equation}
\left\lbrace \begin{array}{c}
\label{eqmcpred2}
 {\dot r}  \ = \ \frac{1}{m}\,p_r,
\\ 
 {\dot p}_r \ = \  -  V^\prime(r),
  \\ 
 {\dot \rho}  \ = \  \frac{1}{m}\frac{\rho}{r}\,p_r\ .
\end{array} \right.  
\end{equation}
Finally, we project the dynamics described by the system (\ref{eqmcpred2}) into the level submanifold $\rho=\alpha$ ($\alpha$ a real constant), and we obtain the system of Hamilton's equations of motion
\begin{equation*}
\left\lbrace \begin{array}{c}
 {\dot r}  \ = \ \frac{1}{m}\,p_r,
\\ 
 {\dot p}_r \ = \  -  V^\prime(r),
\end{array} \right.  
\end{equation*}
which corresponds to the Hamiltonian system with Hamiltonian function
\begin{equation*}
\frac{1}{2\,m}\, p_{r}^2  \ + \ V(r) \ .
\end{equation*} 
\end{example}

Let us see another interesting example. 

\begin{example}
Let us consider again the four-body problem of example \ref{exfourbody}. We have shown that the scalar quantities $rp_{13}=\textbf{r}_{13}\cdot \textbf{p}_{13}$ and $l_{13}=\textbf{p}_ {13}^{2}-w^{2}\textbf{r}_{13}^{2}$ are independent particular integrals; however, they are not in involution, so we cannot perform the projection of the dynamics as in the previous example. 

On the other hand, let us now consider the vector quantities $\textbf{f}=\textbf{p}_{24}-w\textbf{r}_{13}$ and $\textbf{g}=w\textbf{r}_{24}+\textbf{p}_{13}$,
\begin{equation*}
\textbf{f}=\frac{1}{\sqrt{2}}\textbf{P}_{1}-\frac{1}{\sqrt{6}}\textbf{P}_{2}-\frac{2}{\sqrt{3}}\textbf{P}_{3}+\frac{w}{\sqrt{2}}\textbf{S}_{1}+\sqrt{\frac{3}{2}}w\textbf{S}_{2}
\end{equation*}
and
\begin{equation*}
\textbf{g}=-\frac{\sqrt{2}}{2}\textbf{P}_{1}-\frac{\sqrt{6}}{2}\textbf{P}_{2}+\frac{\sqrt{2}}{2}w\textbf{S}_{1}-\frac{\sqrt{6}}{6}w\textbf{S}_{2}+\frac{2\sqrt{3}}{3}w\textbf{S}_{3}.
\end{equation*}
We have
\begin{equation*}
\dot{\textbf{f}}=-w\textbf{g}\quad\textit{and}\quad\dot{\textbf{g}}=w\textbf{f};
\end{equation*}
So, they are independent vector particular integrals (their components are real functions that are particular integrals.) If they were in involution, then we could perform the projection of the dynamics into a Hamiltonian system with the number of degrees of freedom reduced by $4$ (two degrees for each vector quantity); however, in this case, we have
\begin{equation*}
\lbrace \textbf{f},\textbf{g}\rbrace=4w,
\end{equation*}
i.e., they are not in involution.

Not everything is lost; we can see that taking $\textbf{F}=\frac{1}{4w}\textbf{f}$ and $\textbf{G}=\textbf{g}$ (or $\textbf{F}=\textbf{f}$ and $\textbf{G}=\frac{1}{4w}\textbf{g}$) we have that 
\begin{equation*}
\lbrace\textbf{F},\textbf{G}\rbrace=1;
\end{equation*}
so we can find a (local) canonical transformation such that $\textbf{F}$ and $\textbf{G}$ are canonical conjugate coordinates. So in this case, we can still perform the projection of the dynamics into a reduced Hamiltonian system, but the number of degrees of freedom is reduced by $2$. 
\end{example}

The previous example shows us that in some cases, it is possible to perform a projection of the dynamics even when we have two particular integrals that are not in involution. Let us study such cases. 

Let $f$ and $g$ be two independent particular integrals of a given Hamiltonian system. It is easy to see that when $\lbrace f,g\rbrace=c$, with $c$ being a nonzero real number, then we can construct two particular integrals $F=\frac{1}{c}f$ and $G=g$ such that $\lbrace F,G\rbrace=1$. Thus, we can find a (local) canonical transformation such that $F$ and $G$ are conjugate canonical coordinates, allowing us to perform the projection of dynamics into a reduced Hamiltonian system with the number of degrees of freedom reduced by $1$.

\subsection{Particular Liouville integrability}

In \cite{escobar2024particular}, particular integrability is formulated in terms of Lie integrability by quadratures. More precisely, if a Hamiltonian system with $n$ degrees of freedom admits $n$ functionally independent particular integrals in particular involution, then the dynamics can be restricted to a dynamically invariant submanifold on which the resulting equations possess enough commuting symmetries to be integrated by quadratures. The earlier Lie-integrability construction does not require or use a nontrivial reduced Hamiltonian system.

\begin{definition}
\label{def:particularLiouville}
We say that a Hamiltonian system with $n$ degrees of freedom is \emph{particularly Liouville integrable} if there exists a set of $k<n$ particular integrals satisfying the hypothesis of Theorem \ref{maintheorem} such that the reduced Hamiltonian system, by means of those particular integrals, is completely Liouville integrable.
\end{definition}

\begin{remark}
Particular Liouville integrability should be distinguished from the particular Lie integrability studied in \cite{escobar2024particular}. In the Lie-integrable case, we have exactly $n$ particular integrals in particular involution. In the Liouville-integrable case introduced here, $k<n$ and the reduction produces a genuine Hamiltonian system; hence, the usual consequences of Liouville integrability apply to the reduced phase space. In particular, under the standard regularity and compactness hypotheses, the reduced dynamics is quasi-periodic on invariant Liouville tori \cite{AKN2006}.
\end{remark}

This notion also places several restricted forms of integrability in a common framework. If the functions $f_i$ are genuine first integrals, then they are particular integrals with $a_i^j=0$. Thus, Liouville integrability on a regular common level set of first integrals can be viewed as a special case of particular Liouville integrability. In this sense, the weak integrability considered in \cite{rosquist1995invariants,pucacco2004non} and the restricted integrability studied in \cite{maciejewski2025non} fit naturally into the present framework whenever the corresponding restriction admits the Hamiltonian quotient described above.

Let us illustrate this notion with an example.

\begin{example}
Let us consider the planar two-body Coulomb problem in a constant magnetic field, i.e., let us consider two Coulomb charges of opposite sign $(m_1,e_1=e>0\,;\,m_2,e_2=-e)$ moving on a plane with an attractive interaction $\frac{e_1\,e_2}{|\bf{r}_1-\bf{r}_2|}$. Additionally, they are subjected to the presence of an external constant magnetic field $\bf{B}$ perpendicular to the plane. The Hamiltonian system has $4$ degrees of freedom, its Hamiltonian function in canonical coordinates is

\begin{equation*}
H = \ \frac{1}{2\,m_1}\left({\bf{p}_1}-\frac{e}{2}{\bf{B}}\times
{\bf{r}_1}\right)^2 \ + \ \frac{1}{2\,m_2}\left({\bf{p}_2}+\frac{e}{2}{\bf{B}}\times
{\bf{r}_2}\right)^2 \ - \  \frac{e^2}{|\bf{r}_1-\bf{r}_2|} \ ,
\end{equation*}

This system possesses three functionally independent constants of motion, namely, $H$ and the two components $K_{x}$ and $K_{y}$ of the Pseudomomentum ${\bf K}={\bf p}_1+{\bf p}_2+\frac{e}{2}{\bf B}\times ({\bf r}_{1}-{\bf r}_2)$; they are in involution (so the system is partially integrable).

We consider the total canonical angular momentum
\[
J_z=(\boldsymbol r_1\times\boldsymbol p_1+\boldsymbol r_2\times\boldsymbol p_2)_z.
\]
We have that
\[
  \{J_z,H\}=0,
  \qquad
  \{J_z,K_x\}=K_y,
  \qquad
  \{J_z,K_y\}=-K_x.
\]
Consequently, $J_z$ is constant along the characteristic leaves and descends to a function $\overline{J}_{z}$ on $\overline{M}$.

We restrict the dynamics to the invariant submanifold defined by $K_{x}=0,K_{y}=0$; then we obtain Hamiltonian dynamics with $2$ degrees of freedom by projecting the restricted dynamics as described in the previous section. Such a Hamiltonian system is completely Liouville integrable since it has two functionally independent constants of motion in involution, namely, its Hamiltonian function $\overline{H}$ and $\overline{J_z}$. We conclude that the original Hamiltonian system is particularly Liouville integrable.
\end{example}

\section{Lifts of natural Hamiltonian systems and particular integrals}

In the previous section, we studied the conditions under which it is possible to project the dynamics of a Hamiltonian system into a reduced Hamiltonian system by means of the existence of particular integrals. In this section, we study lifts of mechanical Hamiltonian systems that model natural Hamiltonian systems with particular integrals, allowing a projection of the dynamics into another natural Hamiltonian system with fewer degrees of freedom.

A lift of a Hamiltonian system is a procedure that maps its dynamics to another system. Usually, the trajectories of the original system are projections of the trajectories of the second one onto an embedded submanifold (not necessarily dynamically invariant) of the phase space. 

The natural phase space for Hamiltonian mechanics is the cotangent bundle $T^{*}Q$ of the configuration space $Q$; it has a natural symplectic manifold structure given by the canonical symplectic form $\omega_{Q}=-d\theta_{Q}$, where $\theta_{Q}$ is the Liouville 1-form\cite{AMRC2008}. 

Let $(Q,g)$ be a semi-Riemannian manifold of dimension $n$.

\begin{definition}\cite{godbillon1969geometrie,AMRC2008,iglesias2023mechanical}
A mechanical Hamiltonian function on $(Q,g)$ is a real smooth function $H:T^{*}Q\longrightarrow \mathbb{R}$ of the form
\begin{equation*}
H=(g^{*})^{hom}+V\circ \pi_{Q},
\end{equation*}
with $V$ as a possible time-dependent smooth function in $Q$, $g^{*}$ the dual of $g$, and $(g^{*})^{hom}:T^{*}Q\longrightarrow \mathbb{R}$ the homogeneous function associated with $g^{*}$ given by $$(g^{*})^{hom}(\alpha_{q})=\frac{1}{2}g^{*}_{q}(\alpha_{q},\alpha_{q})\quad \forall \alpha_{q}\in T_{q}^{*}Q.$$ 
Hamiltonian systems with mechanical Hamiltonian functions are called natural Hamiltonian systems \cite{Arnold78}.
\end{definition}
The homogeneous function $(g^{*})^{hom}$ represents the kinetic energy function on the configuration space $Q$, and $V$ is the potential energy function; thus, a mechanical Hamiltonian function represents the total energy of a physical system.

Let $(q^{1},\ldots,q^{n},p_{1},\ldots, p_{n})$ be local canonical coordinates on $T^{*}Q$, i.e., $\omega_Q=dq^{i}\wedge dp_{i}$; Suppose that the local expression of $g$ is 
$$g=g_{ij}dq^{i}\otimes dq^{j}$$
with smooth functions $g_{ij}$ on a neighborhood of $Q.$ Then,  the local expression of the dual metric $g^*$ is given by 
\begin{equation*}
g^{*}=g^{ij}\frac{\partial}{\partial x^{i}}\otimes\frac{\partial}{\partial x^{j}},
\end{equation*}
where $(g^{ij})$ is the inverse matrix of $(g_{ij})$. 
Therefore, 
\begin{equation*}
(g^{*})^{hom}(q^{1},\ldots,q^{n},p_{1},\ldots, p_{n})=\frac{1}{2}g^{ij}p_{i}p_{j}.
\end{equation*}

The expression of a mechanical Hamiltonian function $H$ in canonical coordinates is
\begin{equation*}
H(q^{1},\ldots,q^{n},p_{1},\ldots, p_{n})=\frac{1}{2}g^{ij}p_{i}p_{j}+V(q^{1},\ldots,q^{n},t);
\end{equation*}
so the Hamiltonian vector field for $H$ is 
\begin{equation*}\label{hamiltonian}
X_H=g^{ij}p_i\frac{\partial }{\partial q^{j}} + (\Gamma_{si}^jg^{sl}p_jp_l-\frac{\partial V}{\partial q^i})\frac{\partial }{\partial p_i}
\end{equation*}
where $\Gamma_{si}^j$ are the Christoffel symbols of the metric $g$.

\subsection{The Eisenhart lift}
In 1928, Eisenhart introduced a lift that realizes the trajectories of natural Hamiltonian systems as projections of geodesics in an extended configuration space \cite{eisenhart1928dynamical}. We briefly recall the construction; see also \cite{szydlowski1998eisenhart,cariglia2015eisenhart,carinena2017superintegrable}.

Let \((Q,g)\) be an \(n\)-dimensional Riemannian manifold. In local coordinates \((q^1,\ldots,q^n)\), write
\[
g=g_{ij}(q)dq^i\otimes dq^j,
\qquad
ds^2=g_{ij}(q)dq^idq^j
\]
when using the symmetric product notation. Consider the natural Hamiltonian system on \(T^*Q\) with Hamiltonian
\begin{equation*}
H(q,p)
=
\frac{1}{2}g^{ij}(q)p_ip_j+V(q).
\end{equation*}
Assume, locally, that \(V\) is nonzero. The Eisenhart lift is defined on \(Q\times\mathbb R\), with auxiliary coordinate \(z\), by the metric
\begin{equation*}
d\sigma^2
=
g_{ij}(q)dq^idq^j+\frac{1}{V(q)}dz^2.
\end{equation*}
The corresponding geodesic Hamiltonian on \(T^*(Q\times\mathbb R)\) is
\begin{equation*}
\widetilde H(q,z,p,p_z)
=
\frac{1}{2}g^{ij}(q)p_ip_j
+
\frac{1}{2}V(q)p_z^2.
\end{equation*}
Since \(z\) is cyclic, \(p_z\) is a first integral of \(\widetilde H\). The auxiliary equations are
\begin{equation*}
\dot z=V(q)p_z,
\qquad
\dot p_z=0.
\end{equation*}
Thus, the pair \(z,p_z\) forms a system of particular integrals in the sense of Definition \ref{particularintegrals}. In particular, the hypersurface
\[
p_z=0
\]
is dynamically invariant, and on it \(z\) is constant. The induced dynamics on the variables \((q^i,p_i)\) is the geodesic dynamics generated by
\begin{equation*}
h(q,p)
=
\frac{1}{2}g^{ij}(q)p_ip_j.
\end{equation*}

On the other hand, on the invariant level \(p_z=\sqrt{2}\), the projected equations for \((q^i,p_i)\) coincide with the Hamilton equations generated by the original Hamiltonian \(H\), since
\[
\widetilde H|_{p_z=\sqrt{2}}
=
\frac{1}{2}g^{ij}(q)p_ip_j+V(q)
=
H.
\]
On this level, \(z\) is generally not constant since
\[
\dot z=\sqrt{2}\,V(q).
\]
Hence, the original dynamics is recovered by projection along the auxiliary coordinate \(z\), rather than by restricting to a submanifold \(z=\mathrm{constant}\). This is the usual geometric role of the Eisenhart lift.

Let \(f\in C^\infty(T^*Q)\) be regarded as a function on \(T^*(Q\times\mathbb R)\) by pullback under the natural projection. A direct computation gives:
\begin{equation*}
\label{eq:eisenhartbracket}
\{f,\widetilde H\}_{T^*(Q\times\mathbb R)}
=
\{f,H\}_{T^*Q}
+
\frac{\partial V}{\partial q^i}
\frac{\partial f}{\partial p_i}
\left(1-\frac{1}{2}p_z^2\right).
\end{equation*}
Therefore, if \(f\) is a first integral of the original Hamiltonian system, then
\begin{equation*}
\{f,\widetilde H\}_{T^*(Q\times\mathbb R)}
=
\frac{\partial V}{\partial q^i}
\frac{\partial f}{\partial p_i}
\left(1-\frac{1}{2}p_z^2\right).
\end{equation*}
Thus \(f\) is conserved on the invariant level \(p_z^2=2\), but it is not necessarily a first integral of the full Eisenhart-lifted system. More precisely, for each regular value \(c\) of \(f\), the pair
\[
f-c,
\qquad
p_z^2-2
\]
forms a system of particular integrals. Indeed,
\begin{equation*}
\frac{d}{dt}(f-c)
=
-\frac{1}{2}
\frac{\partial V}{\partial q^i}
\frac{\partial f}{\partial p_i}
\left(p_z^2-2\right),
\qquad
\frac{d}{dt}(p_z^2-2)=0.
\end{equation*}
Consequently, the common level set
\[
\{f=c,\ p_z^2=2\}
\]
is dynamically invariant for the Eisenhart-lifted dynamics.
\begin{example}
\label{ex:eisenhartharmonicoscillator}
Consider the two-dimensional planar isotropic harmonic oscillator
\[
H(x,y,p_x,p_y)
=
\frac{1}{2}(p_x^2+p_y^2)
+
\frac{1}{2}(x^2+y^2)
\]
on \(T^*\mathbb R^2\). This system is superintegrable; independent first integrals are
\[
H,\qquad
L_z=xp_y-yp_x,\qquad
h_x=p_x^2+x^2 .
\]

With the normalization used above, the Eisenhart lift is
\[
\widetilde H
=
\frac{1}{2}(p_x^2+p_y^2)
+
\frac{1}{4}(x^2+y^2)p_z^2 .
\]
Indeed, since \(V=\frac12(x^2+y^2)\), the lifted term is
\[
\frac12 Vp_z^2
=
\frac14(x^2+y^2)p_z^2 .
\]
The level \(p_z^2=2\) recovers the original Hamiltonian:
\[
\widetilde H|_{p_z^2=2}
=
\frac{1}{2}(p_x^2+p_y^2)
+
\frac{1}{2}(x^2+y^2)
=
H.
\]

A direct computation gives
\[
\{H,\widetilde H\}_{T^*(Q\times\mathbb R)}
=
(xp_x+yp_y)
\left(1-\frac12p_z^2\right),
\]
\[
\{L_z,\widetilde H\}_{T^*(Q\times\mathbb R)}
=
0,
\]
and
\[
\{h_x,\widetilde H\}_{T^*(Q\times\mathbb R)}
=
2xp_x
\left(1-\frac12p_z^2\right).
\]
Thus \(L_z\) is a genuine first integral of the Eisenhart lift. By contrast, \(H\) and \(h_x\) are conserved on the invariant level
\[
p_z^2=2,
\]
but they are not generally first integrals of the full lifted system.

More precisely, for each regular value \(c\), the pairs
\[
H-c,\qquad p_z^2-2,
\]
and
\[
h_x-c,\qquad p_z^2-2
\]
form systems of particular integrals of the Eisenhart-lifted system. Therefore, the corresponding common level sets
\[
\{H=c,\ p_z^2=2\},
\qquad
\{h_x=c,\ p_z^2=2\}
\]
are dynamically invariant.
\end{example}

It is well known that the Eisenhart lift maps polynomial constants of motion of the original system to polynomial constants of motion of the lifted geodesic system. In fact, given a polynomial constant of motion $g=\displaystyle{\sum_{k=1}^{m}}g_{k}$, with $g_{k}$ being the homogeneous component of degree $k$ in the momenta, the Eisenhart lift of $g$ is
\begin{equation*}
\tilde{g}=\displaystyle{\sum_{k=1}^{m}}g_{k}\left(\frac{p_z}{\sqrt{2}}\right)^{m-k};
\end{equation*}
$\tilde{g}$ is a constant of motion of $\tilde{H}$ \cite{cariglia2014hidden}.

\begin{example}
The harmonic oscillator previously considered provides an example of an Eisenhart lift for a system defined on $\mathbb R^2$ with the Euclidean metric. An interesting non-Euclidean example is given by the so-called two-dimensional Darboux III Hamiltonian system 
    \begin{equation*}
H_\lambda(x,y,p_x,p_y) = \frac{p_x^2+p_y^2}{2 (1+ \lambda (x^2 + y^2))} + \frac{x^2+y^2}{2 (1+ \lambda (x^2 + y^2))} ,
    \end{equation*}
where $\lambda$ is a positive parameter. This Hamiltonian describes the motion of a particle moving on a conformally flat surface with metric 
\begin{equation*}
    d\sigma^2 = (1+ \lambda (x^2 + y^2)) (dx^2 + dy^2) .
\end{equation*}
This system (together with its $N$-dimensional generalization and quantum versions, see \cite{BALLESTEROS2008505,BALLESTEROS20111431,BALLESTEROS2023133618,baena2026fock}) is superintegrable, since 
\begin{equation*}
    H_\lambda, \qquad
C^{(2)}=xp_y-yp_x,\qquad
I_{xx}=p_x^2-(2 \lambda H_\lambda- 1) x^2 ,
\end{equation*}
are functionally independent first integrals. Thus,
\begin{equation*}
    V(x,y) = \frac{x^2+y^2}{2 (1+ \lambda (x^2 + y^2))}
\end{equation*}
and therefore the metric on $Q \times \mathbb R$ defining the lift is given by
\begin{equation*}
    d\sigma^2 = (1+ \lambda (x^2 + y^2)) (dx^2 + dy^2) + 2 \left(\lambda + \frac{1}{x^2+y^2} \right) dz^2.
\end{equation*}
Finally, the Eisenhart lifted Hamiltonian is given by
\begin{equation*}
    \widetilde H_\lambda(x,y,z,p_x,p_y,p_z) = \frac{p_x^2+p_y^2}{2 (1+ \lambda (x^2 + y^2))} + \frac{x^2+y^2}{4 (1+ \lambda (x^2 + y^2))} p_z^2 .
\end{equation*}
Similarly to the harmonic oscillator case, the original Hamiltonian is recovered by the level $p_z^2=2$ of the lifted Hamiltonian, i.e. $\widetilde{H}_\lambda|_{p_z^2=2} = H_\lambda$. Moreover, we have
\begin{equation*}
    \{H_\lambda,\widetilde H_\lambda\}_{T^*(Q\times\mathbb R)} = \left(1-\frac12p_z^2\right) \frac{xp_x+yp_y}{(1+ \lambda (x^2 + y^2))^3},
\end{equation*}
and
\begin{equation*}
    \{I_{xx},\widetilde H_\lambda\}_{T^*(Q\times\mathbb R)} = 2 \left(1-\frac12p_z^2\right) \frac{x (p_x - \lambda y (x p_y - y p_x))}{(1+ \lambda (x^2 + y^2))^3},
\end{equation*}
while 
\begin{equation*}
    \{C^{(2)},\widetilde H_\lambda\}_{T^*(Q\times\mathbb R)} = 0 .
\end{equation*}
Thus, similarly to the harmonic oscillator, \(C^{(2)}\) is a genuine first integral of the Eisenhart lift while \(H_\lambda\) and \(I_{xx}\) are only conserved on the invariant level
\[
p_z^2=2,
\]
and thus they are not first integrals of the full lifted system. Note that the limit $\lambda \to 0$ recovers the results of the previous example, and thus this system can be seen as a nonlinear deformation of the Harmonic oscillator.
\end{example}

\subsection{Lifts producing particular integrals}
\label{secother}
The Eisenhart lift maps appropriately homogenized polynomial first integrals to lifted first integrals. In this subsection, we study other lifts that do not necessarily map the constants of motion of the base system into the constants of motion of the lift. They lift completely integrable Hamiltonian systems into particular integrable ones that are not completely integrable a priori. 

Let us consider again a natural Hamiltonian system on $T^{*}Q$ with the Hamiltonian function in local canonical coordinates given by
\begin{equation*}
H(q^{1},\ldots,q^{n},p_{1},\ldots, p_{n})=\frac{1}{2}g^{ij}(q^{1},\ldots,q^{n})p_{i}p_{j}+V(q^{1},\ldots,q^{n}).
\end{equation*}

\subsubsection{Scalar auxiliary lift}

Let
\[
A\in C^\infty(Q\times \mathbb R)
\]
be a positive smooth function. We define the scalar auxiliary lift on \(T^*(Q\times\mathbb R)\) by
\begin{equation*}
\label{eqlift2}
\widetilde H_{Q\times\mathbb R}(q,z,p,p_z)
=
\frac{1}{2}g^{ij}(q)p_i p_j
+
V(q)
+
A(q,z)p_z^2 ,
\end{equation*}
where \(z\), called the auxiliary variable, denotes the global coordinate on \(\mathbb R\). The function \(A(q,z)\) has a direct mechanical interpretation. It defines the kinetic geometry of the auxiliary direction and plays the role of a position-dependent inverse inertia for the added coordinate \(z\). Indeed, the auxiliary term may be written as
\[
A(q,z)p_z^2=\frac{p_z^2}{2m_{\mathrm{eff}}(q,z)},
\qquad
m_{\mathrm{eff}}(q,z)=\frac{1}{2A(q,z)}.
\]
Hence, \(A\) is not an additional potential for the original system but a warping factor of the extended configuration-space metric. Its dependence on \(q\) determines the coupling between the auxiliary direction and the physical variables, while its dependence on \(z\) controls the evolution of the auxiliary momentum. The assumption \(A>0\) guarantees positive kinetic energy in the added direction. On the invariant hypersurface \(p_z=0\), this auxiliary kinetic contribution vanishes, \(z\) is constant, and the projected dynamics coincides with the original Hamiltonian flow. Away from this hypersurface, \(A\) quantifies how the lifted dynamics departs from the original system.

Equivalently, this is the mechanical Hamiltonian associated with the metric
\[
\sigma
=
g_{ij}(q)dq^i dq^j+\frac{1}{2A(q,z)}dz^2.
\]
If one instead uses the metric coefficient \(1/A\), then the corresponding Hamiltonian contains the term \(\frac12 A p_z^2\); the two conventions differ only by a normalization of \(A\). We assume that the lift is globally defined.

The auxiliary equations are
\begin{equation*}
\dot z=2A(q,z)p_z,
\end{equation*}
and
\begin{equation*}
\dot p_z
=
-\frac{\partial A}{\partial z}(q,z)p_z^2.
\end{equation*}
Hence \(p_z\) is a particular integral since
\begin{equation*}
\dot p_z
=
\left(
-\frac{\partial A}{\partial z}(q,z)p_z
\right)p_z.
\end{equation*}
Moreover, the pair \(z,p_z\) forms a system of particular integrals in the sense of Definition \ref{particularintegrals} because
\[
\dot z=2A(q,z)p_z,
\qquad
\dot p_z=
-\frac{\partial A}{\partial z}(q,z)p_z^2.
\]
Therefore, the hypersurface
\[
M_z=\{p_z=0\}
\]
is dynamically invariant. On \(M_z\), one has \(\dot z=0\), and the equations for \((q^i,p_i)\) reduce to the Hamilton equations generated by \(H\). Thus the lifted system
\[
\left(T^*(Q\times\mathbb R),\omega_{Q\times\mathbb R},\widetilde H_{Q\times\mathbb R}\right)
\]
projects onto the original system
\[
(T^*Q,\omega_Q,H)
\]
after restriction to \(p_z=0\).

Let \(f\in C^\infty(T^*Q)\), regarded also as a function on \(T^*(Q\times\mathbb R)\) by pullback. Then
\begin{equation*}
\label{eq:scalarliftbracket}
\{f,\widetilde H_{Q\times\mathbb R}\}_{T^*(Q\times\mathbb R)}
=
\{f,H\}_{T^*Q}
-
\frac{\partial A}{\partial q^i}
\frac{\partial f}{\partial p_i}
p_z^2 .
\end{equation*}
Consequently, if \(f\) is a first integral of the original Hamiltonian system, then
\begin{equation*}
\{f,\widetilde H_{Q\times\mathbb R}\}_{T^*(Q\times\mathbb R)}
=
-
\frac{\partial A}{\partial q^i}
\frac{\partial f}{\partial p_i}
p_z^2 .
\end{equation*}
Thus \(f\) is conserved on the invariant hypersurface \(p_z=0\), but it is not necessarily a first integral of the full lifted system.

More precisely, for each regular value \(c\) of \(f\), the pair
\[
f-c,\qquad p_z
\]
is a system of particular integrals of the lifted system. Indeed,
\begin{equation*}
\frac{d}{dt}(f-c)
=
-
\frac{\partial A}{\partial q^i}
\frac{\partial f}{\partial p_i}
p_z^2
=
\left(
-
\frac{\partial A}{\partial q^i}
\frac{\partial f}{\partial p_i}
p_z
\right)p_z,
\end{equation*}
while
\[
\dot p_z
=
\left(
-\frac{\partial A}{\partial z}p_z
\right)p_z.
\]
Hence, the common zero set
\[
\{f=c,\ p_z=0\}
\]
is dynamically invariant.

In particular, if \(g(q,p)\) is a polynomial first integral of the original system, then \(g-c\) and \(p_z\) form a system of polynomial particular integrals of the scalar auxiliary lift. A general Eisenhart-type formula producing a modified polynomial first integral of \(\widetilde H_{Q\times\mathbb R}\) requires additional conditions on \(A\) and on the tensorial structure of \(g\); it is not automatic for the scalar auxiliary lift.

\begin{example}
\label{ex:scalarauxiliarylift}
Consider the one-dimensional harmonic oscillator
\[
H(q,p)=\frac{1}{2}p^2+\frac{1}{2}\omega^2q^2
\]
on $T^*\mathbb R$. Let
\[
A(q,z)=1+\varepsilon e^{-q^2}\cos z,
\qquad 0<\varepsilon<1.
\]
The bound on \(\varepsilon\) ensures \(A>0\), so the added kinetic term is positive. Mechanically,
\[
A(q,z)p_z^2
=
\frac{p_z^2}{2m_{\mathrm{eff}}(q,z)},
\qquad
m_{\mathrm{eff}}(q,z)=\frac{1}{2A(q,z)} .
\]
Thus \(A\) acts as a configuration-dependent inverse inertia for the auxiliary degree of freedom. The factor \(e^{-q^2}\) localizes the coupling near the oscillator center, while \(\cos z\) makes the auxiliary inertia periodic along the added direction. On the invariant sector \(p_z=0\), the auxiliary kinetic term vanishes, \(z\) is frozen, and the projected motion is exactly the original harmonic oscillator. Away from this sector, \(A\) controls the geometric departure of the lifted dynamics from the physical oscillator.
The scalar auxiliary lift is
\[
\widetilde H(q,z,p,p_z)
=
\frac{1}{2}p^2+\frac{1}{2}\omega^2q^2
+
A(q,z)p_z^2 .
\]
The lifted equations for the auxiliary variables are
\[
\dot z=2A(q,z)p_z,
\qquad
\dot p_z=-\frac{\partial A}{\partial z}p_z^2
=
\varepsilon e^{-q^2}\sin z\,p_z^2.
\]
Hence $p_z$ is a particular integral:
\[
\dot p_z=
\left(\varepsilon e^{-q^2}\sin z\,p_z\right)p_z.
\]
Therefore the hypersurface
\[
M_{p_z}=\{p_z=0\}
\]
is dynamically invariant. On this hypersurface, the lifted dynamics reduces to
\[
\dot q=p,
\qquad
\dot p=-\omega^2q,
\qquad
\dot z=0.
\]
Thus, after projection along the auxiliary coordinate $z$, one recovers the original harmonic oscillator.

Moreover, the original oscillator energy
\[
I(q,p)=\frac{1}{2}p^2+\frac{1}{2}\omega^2q^2
\]
is not generally conserved by the lifted dynamics. Indeed,
\[
\dot I
=
-\frac{\partial A}{\partial q}p\,p_z^2
=
2\varepsilon q p e^{-q^2}\cos z\,p_z^2.
\]
Thus $I$ is conserved on the invariant hypersurface $p_z=0$, but it is not a first integral of the full lifted system. More precisely, for each regular value $h$ of $I$, the pair
\[
I-h,\qquad p_z
\]
forms a system of particular integrals. Indeed,
\[
\frac{d}{dt}(I-h)
=
2\varepsilon qpe^{-q^2}\cos z\,p_z^2
=
\left(2\varepsilon qpe^{-q^2}\cos z\,p_z\right)p_z,
\]
while
\[
\dot p_z=
\left(\varepsilon e^{-q^2}\sin z\,p_z\right)p_z.
\]
Hence the common zero set
\[
\{I=h,\ p_z=0\}
\]
is dynamically invariant. This illustrates the basic effect of the scalar auxiliary lift: first integrals of the original Hamiltonian system become part of systems of particular integrals of the lifted system, while the original dynamics is recovered by restriction and projection.
\end{example}

\subsubsection{Diagonal auxiliary lift}

We now consider the diagonal multi-auxiliary version of the previous construction. Let
\[
A_1,\ldots,A_r\in C^\infty(Q\times\mathbb R^r)
\]
be positive smooth functions. We define a Riemannian metric on \(Q\times\mathbb R^r\) by
\begin{equation*}
d\sigma^2
=
g_{ij}(q)dq^idq^j
+
\sum_{s=1}^{r}
\frac{1}{2A_s(q,z)}(dz^s)^2.
\end{equation*}
Equivalently, the corresponding mechanical Hamiltonian on \(T^*(Q\times\mathbb R^r)\) is
\begin{equation*}
\widetilde H_{Q\times\mathbb R^r}(q,z,p,p_z)
=
\frac{1}{2}g^{ij}(q)p_ip_j
+
V(q)
+
\sum_{s=1}^{r}A_s(q,z)p_{z_s}^{2}.
\end{equation*}
Here
\[
(q,z,p,p_z)
=
(q^1,\ldots,q^n,z^1,\ldots,z^r,p_1,\ldots,p_n,p_{z_1},\ldots,p_{z_r})
\]
are canonical coordinates on \(T^*(Q\times\mathbb R^r)\).

The auxiliary equations are
\begin{equation*}
\dot z^s
=
2A_s(q,z)p_{z_s},
\qquad s=1,\ldots,r,
\end{equation*}
and
\begin{equation*}
\dot p_{z_s}
=
-
\sum_{l=1}^{r}
\frac{\partial A_l}{\partial z^s}(q,z)p_{z_l}^{2},
\qquad s=1,\ldots,r.
\end{equation*}
Therefore, the auxiliary momenta \(p_{z_1},\ldots,p_{z_r}\) form a system of particular integrals. Indeed,
\begin{equation*}
\dot p_{z_s}
=
b_s^{\,l}p_{z_l},
\qquad
b_s^{\,l}
=
-
\frac{\partial A_l}{\partial z^s}(q,z)p_{z_l},
\end{equation*}
with a summation over \(l\). Hence, the submanifold
\[
M_z=\{p_{z_1}=0,\ldots,p_{z_r}=0\}
\]
is dynamically invariant. On \(M_z\), one has
\[
\dot z^s=0,\qquad s=1,\ldots,r,
\]
and the equations for \((q^i,p_i)\) reduce to the Hamilton equations generated by
\[
H(q,p)=\frac{1}{2}g^{ij}(q)p_ip_j+V(q).
\]
Thus, the lifted system
\[
\left(T^*(Q\times\mathbb R^r),\omega_{Q\times\mathbb R^r},
\widetilde H_{Q\times\mathbb R^r}\right)
\]
projects onto the original system
\[
(T^*Q,\omega_Q,H)
\]
after restriction to \(M_z\).

Let \(f\in C^\infty(T^*Q)\) be regarded as a function on \(T^*(Q\times\mathbb R^r)\) by pullback. Then
\begin{equation*}
\{f,\widetilde H_{Q\times\mathbb R^r}\}_{T^*(Q\times\mathbb R^r)}
=
\{f,H\}_{T^*Q}
-
\sum_{s=1}^{r}
\frac{\partial A_s}{\partial q^j}
\frac{\partial f}{\partial p_j}
p_{z_s}^{2}.
\end{equation*}
Consequently, if \(f\) is a first integral of the original Hamiltonian system, then
\begin{equation*}
\{f,\widetilde H_{Q\times\mathbb R^r}\}_{T^*(Q\times\mathbb R^r)}
=
-
\sum_{s=1}^{r}
\frac{\partial A_s}{\partial q^j}
\frac{\partial f}{\partial p_j}
p_{z_s}^{2}.
\end{equation*}
Thus \(f\) is conserved on \(M_z\), but it is not necessarily a first integral of the full lifted system. More precisely, for each regular value \(c\) of \(f\), the family
\[
f-c,\qquad p_{z_1},\ldots,p_{z_r}
\]
is a system of particular integrals of the lifted system. Hence, the common zero set
\[
\{f=c,\ p_{z_1}=0,\ldots,p_{z_r}=0\}
\]
is dynamically invariant.

\subsubsection{Non-diagonal auxiliary lift}

We now allow the auxiliary metric to have non-diagonal components. Let
\[
(a_{sl}),\qquad s,l=1,\ldots,r,
\]
be a symmetric matrix (at each point) of smooth functions on \(Q\times\mathbb R^r\). We assume that it is globally invertible, and we denote its inverse by
\[
(a^{sl}).
\]
We define a semi-Riemannian metric on \(Q\times\mathbb R^r\) by
\begin{equation*}
d\sigma^2
=
g_{ij}(q)dq^idq^j
+
a_{sl}(q,z)dz^s dz^l.
\end{equation*}
The corresponding mechanical Hamiltonian on \(T^*(Q\times\mathbb R^r)\) is
\begin{equation*}
\widetilde H_{Q\times\mathbb R^r}(q,z,p,p_z)
=
\frac{1}{2}g^{ij}(q)p_ip_j
+
V(q)
+
\frac{1}{2}a^{sl}(q,z)p_{z_s}p_{z_l}.
\end{equation*}
As in the previous constructions, we assume that the lift is globally defined.

The auxiliary equations are
\begin{equation*}
\dot z^s
=
a^{sl}(q,z)p_{z_l},
\qquad s=1,\ldots,r,
\end{equation*}
and
\begin{equation*}
\dot p_{z_s}
=
-
\frac{1}{2}
\frac{\partial a^{kl}}{\partial z^s}(q,z)
p_{z_k}p_{z_l},
\qquad s=1,\ldots,r.
\end{equation*}
Therefore, the auxiliary momenta \(p_{z_1},\ldots,p_{z_r}\) form a system of particular integrals. Indeed,
\begin{equation*}
\dot p_{z_s}
=
b_s^{\,j}p_{z_j},
\qquad
b_s^{\,j}
=
-
\frac{1}{2}
\frac{\partial a^{jl}}{\partial z^s}(q,z)p_{z_l}.
\end{equation*}
Hence
\[
M_z=\{p_{z_1}=\cdots=p_{z_r}=0\}
\]
is dynamically invariant. On \(M_z\), one has
\[
\dot z^s=0,\qquad s=1,\ldots,r,
\]
and the equations for \((q^i,p_i)\) reduce to the Hamilton equations generated by
\[
H(q,p)=\frac{1}{2}g^{ij}(q)p_ip_j+V(q).
\]
Thus, the lifted system
\[
\left(T^*(Q\times\mathbb R^r),\omega_{Q\times\mathbb R^r},
\widetilde H_{Q\times\mathbb R^r}\right)
\]
projects onto the original Hamiltonian system
\[
(T^*Q,\omega_Q,H)
\]
after restriction to \(M_z\).

Let \(f\in C^\infty(T^*Q)\), again regarded as a function on \(T^*(Q\times\mathbb R^r)\) by pullback. Then
\begin{equation*}
\{f,\widetilde H_{Q\times\mathbb R^r}\}_{T^*(Q\times\mathbb R^r)}
=
\{f,H\}_{T^*Q}
-
\frac{1}{2}
\frac{\partial a^{sl}}{\partial q^j}
\frac{\partial f}{\partial p_j}
p_{z_s}p_{z_l}.
\end{equation*}
Consequently, if \(f\) is a first integral of the original Hamiltonian system, then
\begin{equation*}
\{f,\widetilde H_{Q\times\mathbb R^r}\}_{T^*(Q\times\mathbb R^r)}
=
-
\frac{1}{2}
\frac{\partial a^{sl}}{\partial q^j}
\frac{\partial f}{\partial p_j}
p_{z_s}p_{z_l}.
\end{equation*}
Thus \(f\) is conserved on \(M_z\), but it is not necessarily a first integral of the full lifted system. More precisely, for each regular value \(c\) of \(f\), the family
\[
f-c,\qquad p_{z_1},\ldots,p_{z_r}
\]
is a system of particular integrals of the lifted system. Hence, the common zero set
\[
\{f=c,\ p_{z_1}=0,\ldots,p_{z_r}=0\}
\]
is dynamically invariant.

\begin{example}
\label{ex:nondiagonalauxiliarylift}
Consider the two-dimensional isotropic oscillator
\[
H(x,y,p_x,p_y)
=
\frac{1}{2}(p_x^2+p_y^2)
+
\frac{\omega^2}{2}(x^2+y^2)
\]
on $T^*\mathbb R^2$. We add two auxiliary coordinates $z^1,z^2$ and consider the non-diagonal auxiliary lift
\[
\widetilde H
=
H
+
\frac{1}{2}a^{sl}(x,y,z)p_{z_s}p_{z_l},
\qquad s,l=1,2,
\]
where the inverse auxiliary metric is chosen as
\[
(a^{sl})
=
\begin{pmatrix}
1+\varepsilon e^{-x^2}\cos z^1
&
\delta e^{-(x^2+y^2)}\sin(z^1-z^2)
\\[2mm]
\delta e^{-(x^2+y^2)}\sin(z^1-z^2)
&
1+\varepsilon e^{-y^2}\cos z^2
\end{pmatrix}.
\]
For $0<\varepsilon+|\delta|<1$, this matrix is smooth, bounded, and nondegenerate. The lifted Hamiltonian is therefore
\[
\widetilde H
=
\frac{1}{2}(p_x^2+p_y^2)
+
\frac{\omega^2}{2}(x^2+y^2)
+
\frac{1}{2}a^{11}p_{z_1}^2
+
a^{12}p_{z_1}p_{z_2}
+
\frac{1}{2}a^{22}p_{z_2}^2 .
\]

The auxiliary momenta satisfy
\[
\dot p_{z_s}
=
-\frac{1}{2}
\frac{\partial a^{kl}}{\partial z^s}
p_{z_k}p_{z_l},
\qquad s=1,2.
\]
Equivalently,
\[
\dot p_{z_s}=b_s^{\,j}p_{z_j},
\qquad
b_s^{\,j}
=
-\frac{1}{2}
\frac{\partial a^{jl}}{\partial z^s}p_{z_l}.
\]
Thus, the set $\lbrace p_{z_1},p_{z_2}\rbrace$ forms a system of particular integrals. Consequently,
\[
M_z=\{p_{z_1}=0,p_{z_2}=0\}
\]
is dynamically invariant. On $M_z$, the auxiliary variables satisfy
\[
\dot z^s=a^{sl}p_{z_l}=0,
\]
and the remaining equations reduce exactly to
\[
\dot x=p_x,\qquad
\dot y=p_y,\qquad
\dot p_x=-\omega^2x,\qquad
\dot p_y=-\omega^2y.
\]
Hence, after projecting away the constant auxiliary coordinates $z^1,z^2$, one recovers the original oscillator.

The lift is nevertheless nontrivial. Away from $M_z$, the original oscillator energy is not conserved:
\[
\frac{dH}{dt}
=
-\frac{1}{2}
\left(
p_x\frac{\partial a^{sl}}{\partial x}
+
p_y\frac{\partial a^{sl}}{\partial y}
\right)
p_{z_s}p_{z_l}.
\]
Thus the original oscillator energy is conserved on the invariant submanifold $M_z$, but it is not generally a first integral of the full lifted system. More precisely, for each regular value $h$ of $H$, the family
\[
H-h,\qquad p_{z_1},\qquad p_{z_2}
\]
is a system of particular integrals. Indeed,
\[
\frac{d}{dt}(H-h)
=
-\frac{1}{2}
\left(
p_x\frac{\partial a^{sl}}{\partial x}
+
p_y\frac{\partial a^{sl}}{\partial y}
\right)
p_{z_s}p_{z_l},
\]
which belongs to the ideal generated by $p_{z_1}$ and $p_{z_2}$. Also,
\[
\dot p_{z_s}
=
-\frac{1}{2}
\frac{\partial a^{kl}}{\partial z^s}
p_{z_k}p_{z_l},
\qquad s=1,2,
\]
belongs to the same ideal. Hence the common zero set
\[
\{H=h,\ p_{z_1}=0,p_{z_2}=0\}
\]
is dynamically invariant.

This example illustrates the characteristic feature of the non-diagonal auxiliary lift: the auxiliary sector couples different auxiliary momenta, but the invariant constraint $p_{z_1}=p_{z_2}=0$ still projects the lifted dynamics onto the original Hamiltonian system.
\end{example}

\subsection{Natural Hamiltonian systems with magnetic interactions}

Magnetic interactions are encoded by vector potentials; see, for example, \cite{hoque2023family,marchesiello2015three}. To avoid confusion with the auxiliary functions used below, we denote the magnetic vector potential as
\[
\mathcal A=\mathcal A_i(q)dq^i .
\]
The Hamiltonian function for describing the system of a charged particle moving in a static electromagnetic field has the local form
\begin{equation}
\label{eq:emhamiltonian}
H_{\mathrm{m}}(q,p)
=
\frac{1}{2}g^{ij}(q)
\bigl(p_i-e\mathcal A_i(q)\bigr)
\bigl(p_j-e\mathcal A_j(q)\bigr)
+V(q),
\end{equation}
where \(e\in\mathbb R\) is the charge of the particle and \(V\) is the scalar potential.

The Eisenhart lift for Hamiltonians with vector potentials is given by \cite{cariglia2015eisenhart}
\begin{equation*}
\label{eq:emeisenhart}
\widetilde H_{\mathrm{E}}
=
\frac{1}{2}g^{ij}(q)
\bigl(p_i-\mathcal A_i(q)p_z\bigr)
\bigl(p_j-\mathcal A_j(q)p_z\bigr)
+
p_z^2V(q).
\end{equation*}
Restricting to a fixed level \(p_z=e\) recovers the magnetic Hamiltonian using the convention in \cite{cariglia2015eisenhart}
\[
H_{\mathrm{m}}
=
\frac{1}{2}g^{ij}(q)
\bigl(p_i-e\mathcal A_i(q)\bigr)
\bigl(p_j-e\mathcal A_j(q)\bigr)
+
e^2V(q).
\]
Equivalently, one may absorb the charge factor into the definition of the scalar potential and use the convention \eqref{eq:emhamiltonian}. As in the natural case without vector potentials, the Eisenhart lift maps polynomial constants of motion of the original system to polynomial constants of motion of the lifted geodesic system \cite{cariglia2014hidden}.

The auxiliary-lift construction of Section \ref{secother} extends directly to this setting. Let
\[
\Lambda\in C^\infty(Q\times\mathbb R)
\]
be a positive smooth function. We define
\begin{equation}
\label{eq:emauxlift}
\widetilde H_{\Lambda}(q,z,p,p_z)
=
\frac{1}{2}g^{ij}(q)
\bigl(p_i-e\mathcal A_i(q)\bigr)
\bigl(p_j-e\mathcal A_j(q)\bigr)
+
V(q)
+
\Lambda(q,z)p_z^2 .
\end{equation}
Then
\[
\dot z=2\Lambda(q,z)p_z,
\qquad
\dot p_z=-\frac{\partial \Lambda}{\partial z}(q,z)p_z^2.
\]
Thus \(p_z\) is a particular integral, and the hypersurface
\[
M_z=\{p_z=0\}
\]
is dynamically invariant. On \(M_z\), the equations for \((q^i,p_i)\) reduce to the Hamilton equations generated by \(H_{\mathrm{m}}\). Therefore, the auxiliary lift \eqref{eq:emauxlift} recovers the original electromagnetic Hamiltonian dynamics by restricting to \(p_z=0\) and projecting along \(z\).

Moreover, if $F(q,p)$ is a first integral of $H_{\mathrm{m}}$, then
\begin{equation*}
\{F, \widetilde{H}_\Lambda\} = \{F, \Lambda\}\, p_z^2
= \bigl(\{F, \Lambda\}\, p_z\bigr)\, p_z,
\end{equation*}
which always belongs to the ideal generated by $p_z$. Consequently, for each
regular value $c$ of $F$, the pair
\begin{equation*}
F - c, \qquad p_z
\end{equation*}
forms a system of particular integrals of the lifted dynamics.

\begin{example}
\label{ex:emauxiliarylift}
Consider a charged particle moving in the plane under a constant magnetic field \(B\) perpendicular to the plane and an isotropic harmonic potential. In the symmetric gauge,
\[
\mathcal A_x=-\frac{B}{2}y,
\qquad
\mathcal A_y=\frac{B}{2}x,
\]
the magnetic Hamiltonian is
\[
H_{\mathrm{m}}(x,y,p_x,p_y)
=
\frac{1}{2}
\left(p_x+\frac{eB}{2}y\right)^2
+
\frac{1}{2}
\left(p_y-\frac{eB}{2}x\right)^2
+
\frac{\omega^2}{2}(x^2+y^2).
\]
Introduce one auxiliary coordinate \(z\) and define
\[
\widetilde H_{\Lambda}
=
H_{\mathrm{m}}
+
\Lambda(x,y,z)p_z^2,
\]
where
\[
\Lambda(x,y,z)
=
1+\varepsilon e^{-(x^2+y^2)}
\left(
\cos z+\frac{x}{1+x^2}
\right),
\qquad
0<\varepsilon<\frac{1}{3}.
\]
Then \(\Lambda\) is smooth, bounded, and nowhere vanishing. The auxiliary equations are
\[
\dot z=2\Lambda(x,y,z)p_z,
\qquad
\dot p_z
=
-\frac{\partial \Lambda}{\partial z}p_z^2
=
\varepsilon e^{-(x^2+y^2)}\sin z\,p_z^2.
\]
Hence \(p_z\) is a particular integral:
\[
\dot p_z
=
\left(
\varepsilon e^{-(x^2+y^2)}\sin z\,p_z
\right)p_z.
\]
Therefore
\[
M_z=\{p_z=0\}
\]
is dynamically invariant. On \(M_z\), one has \(\dot z=0\), and the equations for
\[
(x,y,p_x,p_y)
\]
reduce exactly to the Hamilton equations generated by \(H_{\mathrm{m}}\). Thus, the original electromagnetic Hamiltonian system is recovered by restricting to \(p_z=0\) and projecting away the auxiliary coordinate \(z\).

The lift is not dynamically trivial. Away from \(M_z\), the original electromagnetic energy is not conserved:
\[
\frac{dH_{\mathrm{m}}}{dt}
=
-\left[
\left(p_x+\frac{eB}{2}y\right)\frac{\partial \Lambda}{\partial x}
+
\left(p_y-\frac{eB}{2}x\right)\frac{\partial \Lambda}{\partial y}
\right]p_z^2.
\]
Thus, \(H_{\mathrm{m}}\) is conserved on the invariant hypersurface \(p_z=0\), but it is not generally a first integral of the full lifted system. More precisely, for each regular value \(h\) of \(H_{\mathrm{m}}\), the pair
\[
H_{\mathrm{m}}-h,\qquad p_z
\]
forms a system of particular integrals. Indeed,
\[
\frac{d}{dt}(H_{\mathrm{m}}-h)
=
-\left[
\left(p_x+\frac{eB}{2}y\right)\frac{\partial \Lambda}{\partial x}
+
\left(p_y-\frac{eB}{2}x\right)\frac{\partial \Lambda}{\partial y}
\right]p_z^2,
\]
which belongs to the ideal generated by \(p_z\), and
\[
\dot p_z
=
\left(
\varepsilon e^{-(x^2+y^2)}\sin z\,p_z
\right)p_z.
\]
Hence
\[
\{H_{\mathrm{m}}=h,\ p_z=0\}
\]
is dynamically invariant.

Moreover, the original planar system is rotationally symmetric, whereas the term
\[
\frac{x}{1+x^2}
\]
in \(\Lambda\) breaks this symmetry in the lifted dynamics. Consequently, the lift is genuinely coupled to the physical degrees of freedom, while still projecting to the original electromagnetic system on the invariant hypersurface. This illustrates that the auxiliary-lift construction extends naturally to Hamiltonian systems with magnetic vector potentials.
\end{example}

\subsection{Polynomial Hamiltonians}

The auxiliary-lift constructions above are not restricted to natural Hamiltonians. The same mechanism can be applied to polynomial Hamiltonians of the form
\[
H(q,p)=\sum_{r=0}^{m}H_r(q,p),
\]
where each \(H_r\) is homogeneous of degree \(r\) in the momenta. Given auxiliary coordinates \(z^1,\ldots,z^s\), consider a lifted Hamiltonian
\[
\widetilde H(q,z,p,p_z)
=
H(q,p)+R(q,z,p_z),
\]
where \(R\) is polynomial in the auxiliary momenta and satisfies
\[
R(q,z,0)=0.
\]
Since $R$ is polynomial in the auxiliary momenta and vanishes when they do,
it can be written as
\begin{equation*}
R(q,z,p_z) = R^i(q,z,p_z)\, p_{z_i},
\end{equation*}
for suitable smooth functions $R^i$, i.e., $R$ belongs to the ideal generated by
$p_{z_1}, \ldots, p_{z_s}$. The same is then true of the partial derivatives
$\partial R/\partial q^j$ and $\partial R/\partial z^i$. Consequently, the
equations for the auxiliary momenta \emph{automatically} close on this ideal:
\begin{equation*}
\dot{p}_{z_i} = \{p_{z_i}, \widetilde{H}\}
= -\frac{\partial R}{\partial z^i}
= b_i^j\, p_{z_j},
\qquad
b_i^j = -\frac{\partial R^j}{\partial z^i},
\qquad i = 1, \ldots, s.
\end{equation*}
Then
\[
p_{z_1},\ldots,p_{z_s}
\]
form a system of particular integrals, and the constraint submanifold
\[
M_z=\{p_{z_1}=0,\ldots,p_{z_s}=0\}
\]
is dynamically invariant. Since \(R(q,z,0)=0\), the restricted dynamics on \(M_z\) projects to the original Hamiltonian dynamics generated by \(H\).

Let \(I(q,p)\) be a polynomial first integral of the original Hamiltonian system. On the lifted phase space,
\[
\{I,\widetilde H\}
=
\{I,H\}+\{I,R\}
=
\{I,R\}.
\]
Thus $I$ need not be a first integral of the full lifted system. However,
since $I$ does not depend on $(z, p_z)$ and $R$ does not depend on $p$,
\begin{equation*}
\{I, R\}
= -\frac{\partial I}{\partial p_j}\,\frac{\partial R}{\partial q^j}
= \left(-\frac{\partial I}{\partial p_j}\,\frac{\partial R^i}{\partial q^j}\right) p_{z_i},
\end{equation*}
which always belongs to the ideal generated by the auxiliary momenta.
Therefore, for each regular value $c$ of $I$, the family
\begin{equation*}
I - c, \qquad p_{z_1}, \ldots, p_{z_s}
\end{equation*}
forms a system of particular integrals of the lifted dynamics, and the common
zero set
\begin{equation*}
\{\, I = c,\ p_{z_1} = 0, \ldots, p_{z_s} = 0 \,\}
\end{equation*}
is dynamically invariant.

Thus, auxiliary lifts of this type always transform polynomial first
integrals of the original system into components of systems of polynomial particular integrals of the lifted system.

\section{Conclusions}

We have introduced a geometric reduction mechanism generated by systems of particular integrals. The starting point is the extension of the scalar condition
\[
\dot f=af
\]
to families of functions satisfying
\[
\dot f_i=a_i^j f_j .
\]
This formulation shows that invariant submanifolds need not be defined by constants of motion, nor by scalar particular integrals alone. They may instead arise as common zero sets of functions whose evolution closes linearly on the family itself. This provides a natural mechanism for restricting a dynamical system to lower-dimensional invariant subsets.

In the Hamiltonian setting, involutive systems of particular integrals carry additional geometric structure. Their common zero set inherits a presymplectic form from the ambient symplectic manifold, and the restricted dynamics projects to a genuine lower-dimensional Hamiltonian system on the quotient manifold by the characteristic distribution of the presymplectic form. This is the central mechanism of the paper: particular integrals define invariant submanifolds, while involutivity turns the restricted dynamics into presymplectic dynamics whose quotient is Hamiltonian.

This framework naturally leads to the notion of particular Liouville integrability. Unlike particular Lie integrability, where the restricted equations need not be Hamiltonian, particular Liouville integrability applies when the invariant restricted projected dynamics is completely integrable in the Liouville sense. Thus, the standard consequences of Liouville integrability, including quasi-periodic motion on invariant tori under the usual compactness and regularity hypotheses, occur not on the full phase space but on the reduced phase space selected by the particular integrals.

We also showed that auxiliary lifts provide a broad source of such systems. In the scalar, diagonal, and non-diagonal auxiliary lifts, the auxiliary momenta define invariant submanifolds on which the lifted dynamics projects to the original Hamiltonian dynamics. Away from these invariant submanifolds, the lifted systems are genuinely coupled, and first integrals of the original dynamics generally cease to be global first integrals of the lifted dynamics, which, together with the auxiliary momenta, form systems of particular integrals. The same mechanism extends to Hamiltonian systems with magnetic vector potentials and to polynomial Hamiltonians.

Several directions remain open. A first one is the classification of auxiliary lifts that transform first integrals into systems of particular integrals. A second direction is the extension of the construction to Poisson, cosymplectic, contact, and locally conformal symplectic systems \cite{deLeon_2024}, and the generalization of these results to time-dependent Hamiltonian systems \cite{ALBERT1989627,GUTIERREZSAGREDO2025105492}. These settings should reveal further ways in which non-global conservation laws organize restricted Hamiltonian dynamics.

\section*{Acknowledgements}

The research of R. Azuaje is supported by the European Union and the Czech Ministry of Education under project CZ.02.01.01/00/22\_011/0008569 "Czech Technical University - International Postdoc Programme CROP".

A. M. Escobar Ruiz would like to acknowledge support from UAM research grant CBI-SA-391-26 PAPDI 2026.

I. Gutierrez-Sagredo acknowledges partial support from the grants PID2023-148373NB-I00 funded by MCIN/AEI/ 10.13039/501100011033/FEDER–UE and BU011P25 funded by Junta de Castilla y León (Spain).

R. Azuaje and I. Gutierrez-Sagredo thank Francisco J. Herranz for valuable discussions at the initial stages of this research. Additionally, R. Azuaje thanks Libor \v{S}nobl for helpful comments on the lifts for Hamiltonian systems with magnetic fields.

\bibliography{refs} 
\bibliographystyle{unsrt} 

\end{document}